\documentclass[10pt,journal,a4paper]{IEEEtran}
\usepackage[normalem]{ulem}
\usepackage{multirow}
\usepackage{caption}
\usepackage{subcaption}
\usepackage{caption}
\usepackage{graphicx}
\usepackage{longtable}
\usepackage{courier}
\usepackage{float}
\usepackage{cite}
\usepackage{mathtools}
\usepackage{amsmath,amssymb}
\usepackage{bbold}
\usepackage{dsfont}
\usepackage[in]{fullpage}
\usepackage[verbose]{wrapfig}
\usepackage{amsthm}
\usepackage{authblk}
\usepackage{xcolor}

\usepackage[linesnumbered]{algorithm2e}

\newcommand{\dfn}{\stackrel{\triangle}{=}}

\def\xv{\boldsymbol{x}}

\def\hv{\boldsymbol{h}}
\def\E{\mathbb{E}}

\def\vv{\boldsymbol{v}}

\def\dv{\boldsymbol{d}}

\def\Uc{\mathcal{U}}
\def\Lc{\boldsymbol{\mathcal{L}}}
\def\Zc{\mathcal{Z}}
\def\Tau{\mathcal{T}}
\def\Tau{\mathcal{T}}

\def\dvlambda{\boldsymbol{d_\lambda}}

\newcommand\smallV{
  \mathchoice
    {{\scriptstyle\mathcal{V}}}
    {{\scriptstyle\mathcal{V}}}
    {{\scriptscriptstyle\mathcal{V}}}
    {\scalebox{.7}{$\scriptscriptstyle\mathcal{O}$}}
  }

\newcommand{\defeq}{\triangleq}

\usepackage[T1]{fontenc}
\usepackage[utf8]{inputenc}
\usepackage{authblk}
\usepackage[margin=0.75in]{geometry}
\usepackage{chngcntr}
\newtheorem{theorem}{Theorem}
\newtheorem{lemma}{Lemma}
\newtheorem{corollary}{Corollary}

\newtheorem{remark}{Remark}

\title{Fundamental Limits of Caching in Heterogeneous Networks with Uncoded Prefetching}

\author{Emanuele Parrinello, Ay\c{s}e \"{U}nsal and Petros~Elia
\thanks{The authors are with the Communication Systems Department at EURECOM, Sophia Antipolis, 06410, France (email: parrinel@eurecom.fr, unsal@eurecom.fr, elia@eurecom.fr). The work is supported by the European Research Council under the EU Horizon 2020 research and innovation program / ERC grant agreement no. 725929 (project DUALITY).}
\thanks{This work is to appear in part in the proceedings of ITW 2018. A preliminary version\nocite{MN14,WanTP15,YuMA16} of this work, focusing only on the single-stream setting, can be found in~\cite{PEU_arxiv_single}.}
}

\begin{document}
\clearpage
\maketitle
\thispagestyle{empty}


\begin{abstract}
The work explores the fundamental limits of coded caching in heterogeneous networks where multiple ($N_0$) senders/antennas, serve different users which are associated (linked) to shared caches, where each such cache helps an arbitrary number of users.  Under the assumption of uncoded cache placement, the work derives the exact optimal worst-case delay and DoF, for a broad range of user-to-cache association profiles where each such profile describes how many users are helped by each cache. This is achieved by presenting an information-theoretic converse based on index coding that succinctly captures the impact of the user-to-cache association, as well as by presenting a coded caching scheme that optimally adapts to the association profile by exploiting the benefits of encoding across users that share the same cache.

The work reveals a powerful interplay between shared caches and multiple senders/antennas, where we can now draw the striking conclusion that, as long as each cache serves at least $N_0$ users, adding a single degree of cache-redundancy can yield a DoF increase equal to $N_0$, while at the same time --- irrespective of the profile --- going from 1 to $N_0$ antennas reduces the delivery time by a factor of $N_0$. Finally some conclusions are also drawn for the related problem of coded caching with multiple file requests.
\end{abstract}


\section{Introduction \label{sec:intro}}
In the context of communication networks, the emergence of predictable content, has brought to the fore the use of caching as a fundamental ingredient for handling the exponential growth in data volumes.
A recent information theoretic exposition of the cache-aided communication problem~\cite{MN14}, has revealed the potential of caching in allowing for the elusive scaling of networks, where a limited amount of (bandwidth and time) resources can conceivably suffice to serve an ever increasing number of users.
\subsection{Coded Caching}
This exposition in~\cite{MN14} considered a shared-link broadcast channel (BC) scenario where a single-antenna transmitter has access to a library of $N$ files, and serves (via a single bottleneck link) $K$ receivers, each having a cache of size equal to the size of $M$ files.

In a normalized setting where the link has capacity 1 file per unit of time, the work in~\cite{MN14} showed that any set of $K$ simultaneous requests (one file requested per user) can be served with a normalized delay (worst-case completion time) which is at most $T = \frac{K(1-\gamma)}{1+K\gamma}$ where $\gamma \defeq \frac{M}{N} $ denotes the normalized cache size. This implied an ability to treat $K\gamma+1$ users at a time; a number that is often referred to as the cache-aided sum \emph{degrees of freedom} (DoF) $d_{\Sigma} \defeq \frac{K(1-\gamma)}{T}$, corresponding to a caching gain of $K\gamma$ additional served users due to caching.

For this same shared-link setting, this performance was shown to be approximately optimal (cf.~\cite{MN14}), and under the basic assumption of uncoded cache placement where caches store uncoded content from the library, it was shown to be exactly optimal (cf.~\cite{WanTP15} as well as \cite{YuMA16}).

Such high coded caching gains have been shown to persist in a variety of settings that include uneven popularity distributions~\cite{JiTLC14,NiesenMtit17Popularity,ZhangLW:18}, uneven topologies~\cite{BidokhtiWT16isit,ZhangE16b}, a variety of channels such as erasure channels~\cite{GhorbelKY:16}, MIMO broadcast channels with fading~\cite{ZE:17tit}, a variety of networks such as D2D networks~\cite{JiCM16D2D}, coded caching under secrecy constraints \cite{RPK+:16}, and in other settings as well~\cite{SenguptaTS15,CaoTXL16,RoigTG17a,CaoTaoMultiAntenna18,PiovClerckISIT18,BayatMC:18arxiv}. \nocite{lampiris2018lowCSIT,LampirisZE17} Recently some progress has also been made in ameliorating the well known subpacketization bottleneck; for this see for example~\cite{ShanmugamJTLD16it,JiSVLTC15,YanCTC:17tit,TangR:17isit,ShangguanZG:18tit,ShanmugamTD:17isit,LampirisEliaJsac18}.


\begin{figure}[t!]
\centering
\includegraphics[width=0.4\linewidth]{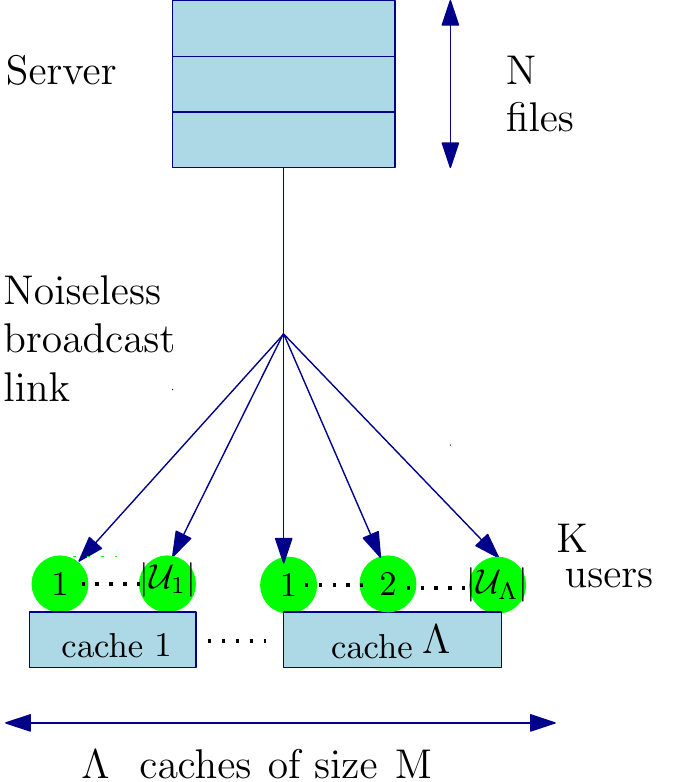}  \hspace{5pt}\includegraphics[width=0.4\linewidth]{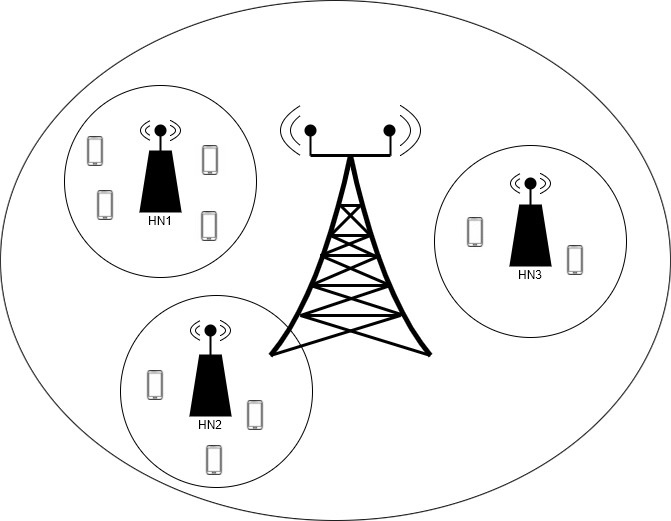}
    \label{system_pic}
  \caption{Shared-link or multi-antenna broadcast channel with shared caches.}
\end{figure}
\subsection{Cache-aided Heterogeneous Networks: Coded Caching with Shared Caches}
Another step in further exploiting the use of caches, was to explore coded caching in the context of the so-called heterogeneous networks which better capture aspects of more realistic settings such as larger wireless networks. Here the term \emph{heterogeneous} refers to scenarios where one or more (typically) multi-antenna transmitters (base-stations) communicate to a set of users, with the assistance of smaller nodes. In our setting, these smaller helper nodes will serve as caches that will be shared among the users. This cache-aided heterogeneous topology nicely captures an evolution into denser networks where many wireless access points work in conjunction with bigger base stations, in order to better handle interference and to alleviate the backhaul load by replacing backhaul capacity with storage capacity at the communicating nodes.

The use of caching in such networks was famously explored in the \textit{Femtocaching} work in~\cite{GSDMC:12}, where wireless receivers are assisted by helper nodes of a limited cache size, whose main role is to bring content closer to the users. A transition to coded caching can be found in~\cite{Diggavi_IT} which considered a similar shared-cache heterogeneous network as here, where each receiving user can have access to a main single-antenna base station (single-sender) and to different helper caches. In this context, under a uniform user-to-cache association where each cache serves an equal number of users, \cite{Diggavi_IT} proposes a coded caching scheme which was shown to perform to within a certain constant factor from the optimal. This uniform setting is addressed also in~\cite{MND13}, again for the single antenna case.
Interesting work can also be found in \cite{XuGW:18sharedArxiv} which explores the single-stream (shared-link) coded caching scenario with shared caches, where the uniformity condition is lifted, and where emphasis is placed on designing schemes with centralized coded prefetching with small sum-size caches where the total cache size is smaller than the library size (i.e., where $KM < N$).

\paragraph*{Coded caching with multiple file requests}
Related work can also be found on the problem of coded caching with multiple file requests per receiver, which --- for the single-stream, error free case --- is closely related to the shared cache problem here. Such work --- all in the shared-link case ($N_0=1$) --- appears in \cite{Ji2015CachingaidedCM,SenguptaT:17TCOMM} in the context of single-layer coded caching. Somewhat related work also appears in~\cite{ZhangWXWL16,KaramchandaniNMD16IT} in the context of hierarchical coded caching. Recent progress can also be found in~\cite{WeUlukus17} which establishes the exact optimal worst-case delay --- under uncoded cache placement, for the shared-link case --- for the uniform case where each user requests an equal number of files\footnote{This work explores other cases as well, such as that where the performance measure is the average delivery delay, for which various bounds are presented.}.  As a byproduct of our results here, in the context of worst-case demands, we establish the exact optimal performance of the multiple file requests problem for any (not necessarily uniform) user-to-file association profile.

\paragraph*{Current work}
In the heterogeneous setting with shared caches, we here explore the effect of user-to-cache association profiles and their non-uniformity, and we characterize how this effect is scaled in the presence of multiple antennas or multiple senders. Such considerations are motivated by realistic constraints in assigning users to caches, where these constraints may be due to topology, cache capacity or other factors. As it turns out, there is an interesting interplay between all these aspects, which is crisply revealed here as a result of a new scheme and an outer bound that jointly provide exact optimality results. Throughout the paper, emphasis will be placed on the shared-cache scenario, but some of the results will be translated directly to the multiple file request problem which will be described later on.

\subsection{Notation}
For $n$ being a positive integer, $[n]$ refers to the following set $[n]\triangleq \{1,2,\dots,n\}$, and $2^{[n]}$ denotes the power set of $[n]$. The expression $\alpha | \beta$ denotes that integer $\alpha$ divides  integer $\beta$. Permutation and binomial coefficients are denoted and defined by $P(n,k)\defeq  \frac{n!}{(n-k)!}$ and $\binom{n}{k}\defeq  \frac{n!}{(n-k)!k!}$, respectively.
For a set $\mathcal{A}$, $|\mathcal{A}|$ denotes its cardinality.
$\mathbb{N}$ represents the natural numbers. We denote the lower convex envelope of the points $\{(i, f(i)) | i \in [n]\cup \{0\}\}$ for some $n\in \mathbb{N}$ by $Conv(f(i))$. The concatenation of a vector $\vv$ with itself $N$ times is denoted by $(\vv \Vert \vv)_{N}$.  
For $n\in \mathbb{N}$, we denote the symmetric group of all permutations of $[n]$ by $S_n$.
To simplify notation, we will also use such permutations $\pi\in S_n$ on vectors $\vv \in \mathbb{R}^n$, where $\pi(\vv)$ will now represent the action of the permutation matrix defined by $\pi$, meaning that the first element of $\pi(\vv)$ is $\vv_{\pi(1)}$ (the $\pi(1)$ entry of $\vv$), the second is $\vv_{\pi(2)}$, and so on. Similarly $\pi^{-1}(\cdot )$ will represent the inverse such function and $\pi_s(\vv)$ will denote the sorted version of a real vector $\vv$ in descending order.

\subsection{Paper Outline}
In Section~\ref{sec:systemModel} we give a detailed description of the system model and the problem definition, followed by the main results in Section~\ref{sec:results}, first for the shared-link setting with shared caches\footnote{We also introduce a very brief parenthetical note that translates these results to the multiple file request scenario.}, and then for the multi-antenna/multi-sender setting. In Section~\ref{sec:scheme}, we introduce the scheme for a broad range of parameters. The scheme is further explained with an example in this section. We present the information theoretic converse along with an explanatory example for constructing the lower bound in Section~\ref{sec:converse}.  Lastly, in Section~\ref{sec:discussion} we draw some basic conclusions, while in the Appendix Section~\ref{sec:Appendix} we present some proof details.


\section{System Model\label{sec:systemModel}}
We consider a basic broadcast configuration with a transmitting server having $N_0$ transmitting antennas and access to a library of $N$ files $W^{1},W^{2},\dots ,W^{N}$, each of size equal to one unit of `file', where this transmitter is connected via a broadcast link to $K$ receiving users and to $\Lambda\leq K$ helper nodes that will serve as caches which store content from the library\footnote{
We note that while the representation here is of a wireless model, the results apply directly to the standard wired multi-sender setting. In the high-SNR regime of interest, when $N_0=1$ and $\Lambda = K$ (where each cache is associated to one user), the setting matches identically the original single-stream shared-link setting in \cite{MN14}. In particular, the file size and log(SNR) are here scaled so that, as in \cite{MN14}, each point-to-point link has (ergodic) capacity of 1 file per unit of time. When $N_0>1 $ and $\Lambda = K$ (again each cache is associated to one user), the setting matches the multi-server \emph{wireline} setting of \cite{ShariatpanahiMK16it} with a fully connected linear network, which we now explore in the presence of fewer caches serving potentially different numbers of users.}.
The communication process is split into $a)$ the cache-placement phase, $b)$ the user-to-cache assignment phase during which each user is assigned to a single cache, and $c)$ the delivery phase where each user requests a single file independently and during which the transmitter aims to deliver these requested files, taking into consideration the cached content and the user-to-cache association.

\paragraph{Cache placement phase}
During this phase, helper nodes store content from the library without having knowledge of the users' requests. Each helper cache has size $M\leq N$ units of file, and no coding is applied to the content stored at the helper caches; this corresponds to the common case of \textit{uncoded cache placement}. We will denote by $\mathcal{Z}_{\lambda}$ the content stored by helper node $\lambda$ during this phase. The cache-placement algorithm is oblivious of the subsequent user-to-cache association $\mathcal{U}$.

\paragraph{User-to-cache association}
After the caches are filled, each user is assigned to exactly \emph{one} helper node/cache, from which it can download content at zero cost. Specifically, each cache $
\lambda = 1,2,\dots,\Lambda$, is assigned to a set of users $\mathcal{U}_\lambda$, and all these disjoint sets \[\mathcal{U}\dfn \{\mathcal{U}_1,\mathcal{U}_2,\dots ,\mathcal{U}_\Lambda\}\] form the partition of the set of users $\{1,2,\dots,K\}$, describing the overall association of the users to the caches.

This cache assignment is independent of the cache content and independent of the file requests to follow. We here consider any arbitrary user-to-cache association $\mathcal{U}$, thus allowing the results to reflect both an ability to choose/design the association, as well as to reflect possible association restrictions due to randomness or topology. Similarly, having the user-to-cache association being independent of the requested files, is meant to reflect the fact that such associations may not be able to vary as quickly as a user changes the requested content.

\paragraph{Content delivery}
The delivery phase commences when each user $k = 1,\dots,K$ requests from the transmitter, any \emph{one} file $W^{d_{k}}$, $d_{k}\in\{1,\dots,N\}$ out of the $N$ library files.
Upon notification of the entire \emph{demand vector} $\dv=(d_1,d_2,\dots,d_{K})\in\{1,\dots,N\}^K$, the transmitter aims to deliver the requested files, each to their intended receiver, and the objective is to design a \emph{caching and delivery scheme $\chi$} that does so with limited (delivery phase) duration $T$, where the delivery algorithm has full knowledge of the user-to-cache association $\mathcal{U}$.

For each transmission, the received signals at user $k$, take the form
\begin{align}
y_{k}=\hv_{k}^{T} \xv + w_{k}, ~~ k = 1, \dots, K
\end{align}
where $\xv\in\mathbb{C}^{N_0\times 1}$ denotes the transmitted vector satisfying a power constraint $\E(||\xv||^2)\leq P$, $\hv_{k}\in\mathbb{C}^{N_0\times 1}$ denotes the channel of user $k$, and $w_{k}$ represents unit-power AWGN noise at receiver $k$. We will assume that the allowable power $P$ is high (i.e., we will assume high signal-to-noise ratio (SNR)), that there exists perfect channel state information throughout the (active) nodes, that fading is statistically symmetric, and that each link (one antenna to one receiver) has ergodic capacity $\log(SNR)+o(log(SNR))$.

\paragraph{User-to-cache association profiles, and performance measure}
As one can imagine, some user-to-cache association instances $\mathcal{U}$ may allow for higher performance than others; for instance, one can suspect that more uniform profiles may be preferable. Part of the objective of this work is to explore the effect of such associations on the overall performance. Toward this, for any given $\mathcal{U}$, we define the association \textit{profile} (sorted histogram)
$$\Lc=(\mathcal{L}_{1},\dots,\mathcal{L}_{\Lambda})$$ where $\mathcal{L}_{\lambda}$ is the number of users assigned to the $\lambda$-th \emph{most populated} helper node/cache\footnote{Here $\Lc$ is simply the vector of the cardinalities of $\mathcal{U}_\lambda, ~\forall\lambda\in\{1,\dots,\Lambda\}$, sorted in descending order. For example, $\mathcal{L}_{1}=6$ states that the highest number of users served by a single cache, is $6$.}. Naturally, $\sum_{\lambda=1}^\Lambda \mathcal{L}_{\lambda} = K$. Each profile $\Lc$ defines a \emph{class} $\mathcal{U}_{\Lc}$ comprising all the user-to-cache associations $\mathcal{U}$ that share the same\footnote{An example of a user-to-cache assignment could have that users $\mathcal{U}_1=(14,15)$ are assigned to helper node $1$, users $\mathcal{U}_2=(1,2,3,4,5,6,7,8)$ are assigned to helper node $2$, and users $\mathcal{U}_3=(9,10,11,12,13)$ to helper node $3$. This corresponds to a profile $\Lc=(8,5,2)$. The assignment $\mathcal{U}_1=(1,3,5,7,9,11,13,15)$, $\mathcal{U}_2=(2,4)$, $\mathcal{U}_3=(6,8,10,12,14)$ would have the same profile, and the two resulting $\mathcal{U}$ would belong to the same class labeled by $\Lc=(8,5,2)$.} profile $\Lc$.

As in \cite{MN14}, the measure of interest $T$ is the number of time slots, per file served per user, needed to complete delivery of any file-request vector\footnote{The time scale is normalized such that one time slot corresponds to the optimal amount of time needed to send a single file from the transmitter to the receiver, had there been no caching and no interference.} $\dv$. We use $T(\mathcal{U},\dv,\chi)$ to define the delay required by some generic caching-and-delivery scheme $\chi$ to satisfy demand $\dv$ in the presence of a user-to-cache association described by $\mathcal{U}$. To capture the effect of the user-to-cache association, we will characterize the optimal worst-case delivery time
\begin{equation}
T^*(\Lc)\defeq \min_{\chi} \max_{(\mathcal{U},\dv) \in (\mathcal{U}_{\Lc},\{1,\dots,N\}^K)} T(\mathcal{U},\dv,\chi)  \label{eq:T*_def}
\end{equation}
for each class. Our interest is in the regime of $N\geq K$ where there are more files than users.



\section{Main Results \label{sec:results}}

We first describe the main results for the single antenna case\footnote{This is also presented in the preliminary version of this work in \cite{PEU_arxiv_single}.} (shared-link BC), and then generalize to the multi-antenna/multi-sender case.

\subsection{Shared-Link Coded Caching with Shared Caches\label{subsec:single_antenna_lower}}

The following theorem presents the main result for the shared-link case ($N_0 = 1$).
\begin{theorem}\label{thm:PerClassSingleAntenna}
In the $K$-user shared-link broadcast channel with $\Lambda$ shared caches of normalized size $\gamma$, the optimal delivery time within any class/profile $\Lc$ is
\begin{equation}\label{eq:TS_L}
T^*(\Lc)=Conv\bigg(\frac{\sum_{r=1}^{\Lambda-\Lambda\gamma}\mathcal{L}_r{\Lambda-r\choose \Lambda\gamma}}{{\Lambda\choose \Lambda\gamma}}\bigg)
\end{equation}
at points $\gamma\in \{\frac{1}{\Lambda},\frac{2}{\Lambda},\dots,1\}$.
\end{theorem}

\vspace{3pt}\emph{Proof.} The achievability part of the proof is given in Section~\ref{sec:scheme}, and the converse is proved in Section~\ref{sec:converse} after setting $N_0 = 1$.\vspace{3pt}

\begin{remark}
We note that the converse that supports Theorem~\ref{thm:PerClassSingleAntenna}, encompasses the class of all caching-and-delivery schemes $\chi$ that employ uncoded cache placement under a general sum cache constraint
$\frac{1}{\Lambda}\sum_{\lambda=1}^\Lambda |\mathcal{Z}_\lambda | = M$
which does not \emph{necessarily} impose an individual cache size constraint. The converse also encompasses all scenarios that involve a library of size $\sum_{n\in[N]}|W^{n}| = N$ but where the files may be of different size. In the end, even though the designed optimal scheme will consider an individual cache size $M$ and equal file sizes, the converse guarantees that there cannot exist a scheme (even in settings with uneven cache sizes or uneven file sizes) that exceeds the optimal performance identified here.
\end{remark}

From Theorem~\ref{thm:PerClassSingleAntenna}, we see that in the uniform case\footnote{Here, this uniform case, naturally implies that $\Lambda|K$.} where $\Lc=(\frac{K}{\Lambda},\frac{K}{\Lambda},\dots,\frac{K}{\Lambda})$, the expression in~\eqref{eq:TS_L} reduces to
\[T^*(\Lc)=\frac{K(1-\gamma)}{\Lambda\gamma+1}\]
matching the achievable delay presented in \cite{MND13}. It also matches the recent result by \cite{WeUlukus17} which proved that this performance --- in the context of the multiple file request problem --- is optimal under the assumption of uncoded cache placement.

The following corollary relates to this uniform case.
\begin{corollary}\label{cor:ressym}
In the uniform user-to-cache association case where $\Lc=(\frac{K}{\Lambda},\frac{K}{\Lambda},\dots,\frac{K}{\Lambda})$, the aforementioned optimal delay $T^*(\Lc)=\frac{K(1-\gamma)}{\Lambda\gamma+1}$ is smaller than the corresponding delay $T^*(\Lc)$ for any other non-uniform class.
\end{corollary}

\begin{proof}
The proof that the uniform profile results in the smallest delay among all profiles, follows directly from the fact that in~\eqref{eq:TS_L}, both $\mathcal{L}_r$ and ${\Lambda-r\choose \Lambda\gamma}$ are non-increasing with $r$.\end{proof}


\subsection{Multi-antenna/Multi-sender Coded Caching with Shared Caches}
The following extends Theorem~\ref{thm:PerClassSingleAntenna} to the case where the transmitter is equipped with multiple ($N_0>1$) antennas. The results hold for any $\Lc$ as long as any non zero $\mathcal{L}_\lambda$ satisfies $\mathcal{L}_\lambda\geq N_0, ~\forall\lambda\in[\Lambda]$.
\begin{theorem}\label{thm:resmultiant}
In the $N_0$-antenna $K$-user broadcast channel with $\Lambda$ shared caches of normalized size $\gamma$, the optimal delivery time within any class/profile $\Lc$ is
\begin{equation}\label{eq:multi_delay}
T^*(\Lc,N_0)=\frac{1}{N_0}Conv\bigg(\frac{\sum_{r=1}^{\Lambda-\Lambda\gamma}\mathcal{L}_{r}{\Lambda-r\choose \Lambda\gamma}}{{\Lambda\choose \Lambda\gamma}}\bigg)\\
\end{equation} for $\gamma\in \left\{\frac{1}{\Lambda},\frac{2}{\Lambda},\dots,1\right\}$. This reveals a multiplicative gain of $N_0$ with respect to the single antenna case.
\end{theorem}
\vspace{3pt}
\emph{Proof.} The scheme that achieves (\ref{eq:multi_delay}) is presented in Section~\ref{sec:scheme}, and the converse is presented in Section~\ref{sec:converse}.
\vspace{3pt}

The following extends Corollary~\ref{cor:ressym} to the multi-antenna case, and the proof is direct from Theorem~\ref{thm:resmultiant}.
\begin{corollary}\label{cor:ressymMulti}
In the uniform user-to-cache association case of $\Lc=\left(\frac{K}{\Lambda},\frac{K}{\Lambda},\dots,\frac{K}{\Lambda}\right)$ where $N_0\leq \frac{K}{\Lambda}$, the optimal delay is
\begin{equation} \label{delay_unif_multi}
T^*(\Lc)=\frac{K(1-\gamma)}{N_0(\Lambda\gamma+1)}
\end{equation}
and it is smaller than the corresponding delay $T^*(\Lc)$ for any other non-uniform class.
\end{corollary}

\begin{remark}[Shared-link coded caching with multiple file requests]\label{rem:multipleFilerequstsResult}
In the error-free shared-link case ($N_0 = 1$), with file-independence and worst-case demand assumptions, the shared-cache problem here is closely related to the coded caching problem with multiple file requests per user, where now $\Lambda$ users with their own cache, request in total $K\geq \Lambda$ files. In particular, changing a bit the format, now each demand vector $\dv = (d_1,d_2,\dots,d_K)$ would represent the vector of the indices of the $K$ requested files, and each user $\lambda = \{1,2,\dots,\Lambda\}$, would request those files from this vector $\dv$, whose indices\footnote{For example, having $\mathcal{U}_2 = \{3,5,7\}$, means that user 2 has requested files $W^{d_{3}},W^{d_{5}},W^{d_{7}}$.} form the set $\mathcal{U}_\lambda \subset [K]$.
At this point, as before, the problem is now defined by the user-to-file association $\mathcal{U} = \{\mathcal{U}_1,\mathcal{U}_2,\dots ,\mathcal{U}_\Lambda\}$ which describes --- given a fixed demand vector $\dv$ --- the files requested by any user. From this point on, the equivalence with the original shared cache problem is complete. As before, each such $\mathcal{U}$ again has a corresponding (sorted) profile $\Lc=(\mathcal{L}_{1},\mathcal{L}_{2},\dots,\mathcal{L}_{\Lambda})$, and belongs to a class $\mathcal{U}_{\Lc}$ with all other associations $\mathcal{U}$ that share the same profile $\Lc$. As we quickly show in the Appendix Section~\ref{sec:AppendixMultipleFileRequests}, our scheme and converse can be adapted to the multiple file request problem, and thus directly from Theorem~\ref{cor:ressym} we conclude that for this multiple file request problem, the optimal delay $T^*(\Lc)\defeq \min_{\chi} \max_{(\mathcal{U},\dv) \in (\mathcal{U}_{\Lc},\{1,\dots,N\}^K)} T(\mathcal{U},\dv,\chi)$ corresponding to any user-to-file association profile $\Lc$, takes the form $T^*(\Lc)= Conv\bigg(\frac{\sum_{r=1}^{\Lambda-\Lambda\gamma}\mathcal{L}_{r}{\Lambda-r\choose \Lambda\gamma}}{{\Lambda\choose \Lambda\gamma}}\bigg)$.  At this point we close the parenthesis regarding multiple file requests, and we refocus exclusively on the problem of shared caches.
\end{remark}

\subsection{Interpretation of Results}
\subsubsection{Capturing the effect of the user-to-cache association profile}
In a nutshell, Theorems~\ref{thm:PerClassSingleAntenna},\ref{thm:resmultiant} quantify how profile non-uniformities bring about increased delays. What we see is that, the more skewed the profile is, the larger is the delay. This is reflected in Figure~\ref{fig:performance} which shows --- for a setting with $K=30$ users and $\Lambda=6$ caches --- the memory-delay trade-off curves for different user-to-cache association profiles. As expected, Figure~\ref{fig:performance} demonstrates that when all users are connected to the same helper cache, the only gain arising from caching is the well known \textit{local caching gain}. On the other hand, when users are assigned uniformly among the caches (i.e., when $\mathcal{L}_{\lambda}=\frac{K}{\Lambda},\forall\lambda\in[\Lambda]$) the caching gain is maximized and the delay is minimized.

\begin{figure}[t!]
\centering
\includegraphics[width=0.95\linewidth]{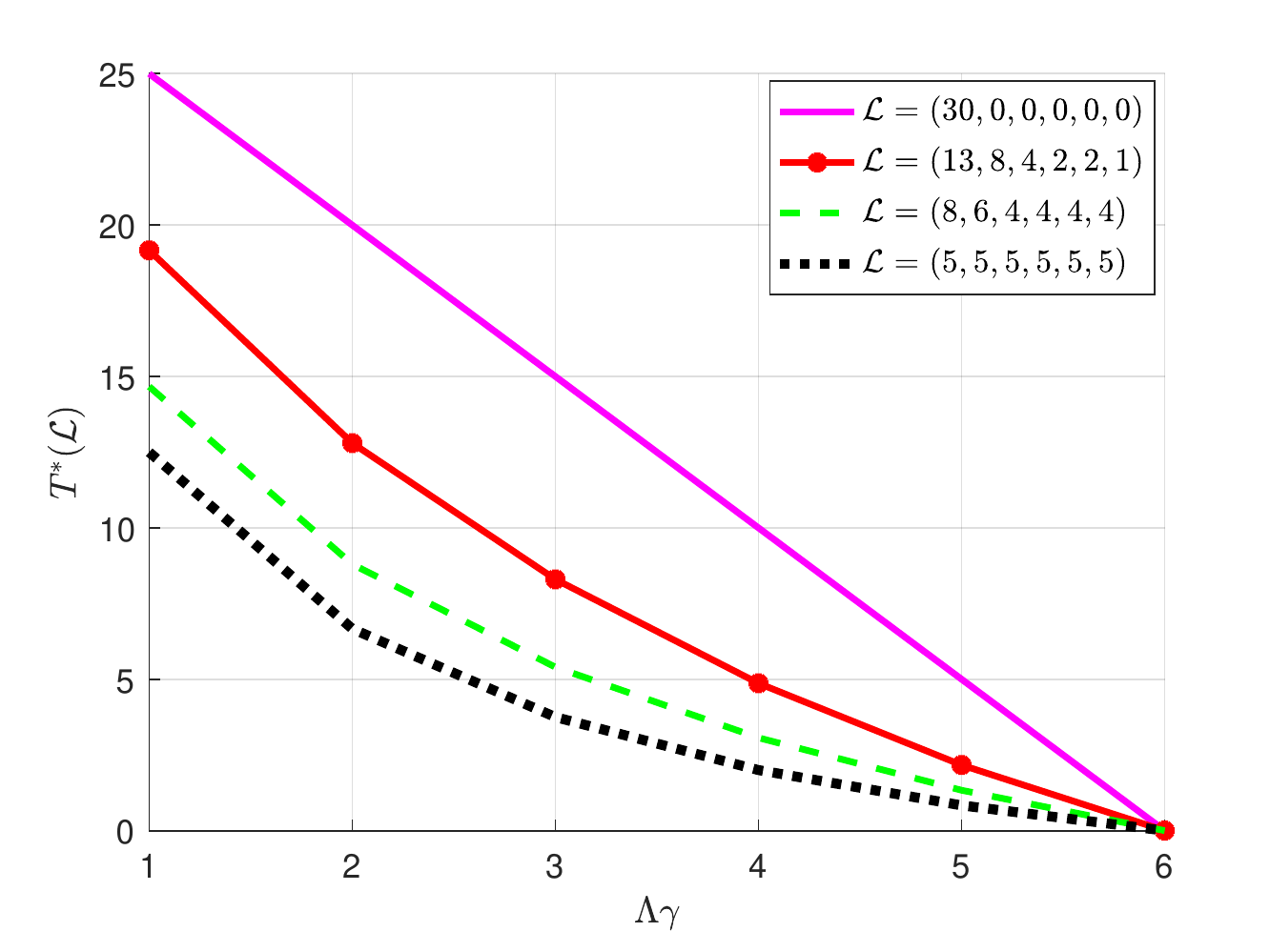}
  \caption{Optimal delay for different user-to-cache association profiles $\Lc$, for $K=30$ users and $\Lambda=6$ caches.}
  \label{fig:performance}
\end{figure}

\subsubsection{A multiplicative reduction in delay}

Theorem~\ref{thm:resmultiant} states that, as long as each cache is associated to at least $N_0$ users, we can achieve a delay $T(\Lc,N_0)= \frac{1}{N_0}\frac{\sum_{r=1}^{\Lambda-\Lambda\gamma}\mathcal{L}_{r}{\Lambda-r\choose \Lambda\gamma}}{{\Lambda\choose \Lambda\gamma}}$. The resulting reduction
\begin{equation}\label{eq:ratioT}
\frac{T(\Lc,N_0=1)}{T(\Lc,N_0) }= N_0
\end{equation}
as compared to the single-stream case, comes in strong contrast to the case of $\Lambda = K$ where, as we know from~\cite{ShariatpanahiMK16it}, this same reduction takes the form
\begin{equation}\label{eq:ratioTold}
\frac{T(\Lambda = K,N_0=1)}{T(\Lambda = K,N_0)} = \frac{\frac{K(1-\gamma)}{1+\Lambda\gamma}}{\frac{K(1-\gamma)}{N_0+\Lambda\gamma}} = \frac{N_0+\Lambda\gamma}{1+\Lambda\gamma}
\end{equation}
which approaches $N_0$ only when $\gamma \rightarrow 0$, and which decreases as $\gamma$ increases.

In the uniform case ($\mathcal{L}_\lambda = \frac{K}{\Lambda}$) with $\Lambda \leq \frac{K}{N_0}$, Corollary~\ref{cor:ressymMulti} implies a sum-DoF \[d_\Sigma(\gamma)= \frac{K(1-\gamma)}{T} = N_0(1+\Lambda \gamma)\] which reveals that every time we add a single degree of cache-redundancy (i.e., every time we increase $\Lambda\gamma$ by one), we gain $N_0$ degrees of freedom. This is in direct contrast to the case of $\Lambda = K$ (for which case we recall from~\cite{ShariatpanahiMK16it} that the DoF is $N_0+\Lambda\gamma$) where the same unit increase in the cache redundancy yields only one additional DoF.

\subsubsection{Impact of encoding over users that share the same cache}
As we know, both the MN algorithm in~\cite{MN14} and the multi-antenna algorithm in \cite{ShariatpanahiMK16it}, are designed for users with different caches, so --- in the uniform case where $\mathcal{L}_\lambda = K/\Lambda$ --- one conceivable treatment of the shared-cache problem would have been to apply these algorithms over $\Lambda$ users at a time, all with different caches\footnote{This would then require $\frac{K}{\Lambda}$ such rounds in order to cover all $K$ users.}. As we see, in the single antenna case, this implementation would treat $1+\Lambda\gamma$ users at a time thus yielding a delay of $T = \frac{K(1-\gamma)}{1+\Lambda\gamma}$, while in the multi-antenna case, this implementation would treat $N_0+\Lambda\gamma$ users at a time (see~\cite{ShariatpanahiMK16it}) thus yielding a delay of $T = \frac{K(1-\gamma)}{N_0+\Lambda\gamma}$. What we see here is that while this direct implementation is optimal (this is what we also do here in the uniform-profile case) in the single antenna case (see \cite{WeUlukus17}, see also Corollary~\ref{cor:ressym}), in the multi-antenna case, this same approach can have an unbounded performance gap
\begin{equation}\label{eq:gapNaive}
  \frac{\frac{K(1-\gamma)}{N_0+\Lambda\gamma}}{\frac{K(1-\gamma)}{N_0(1+\Lambda\gamma)}} = \frac{N_0(1+\Lambda\gamma)}{N_0+\Lambda\gamma}
\end{equation}
from the derived optimal performance from Corollary~\ref{cor:ressymMulti}. These conclusions also apply when the user-to-cache association profiles are not uniform; again there would be a direct implementation of existing multi-antenna coded caching algorithms, which would though again have an unbounded performance gap from the optimal performance achieved here.

\section{Coded Caching Scheme\label{sec:scheme}}

This section is dedicated to the description of the placement-and-delivery scheme achieving the performance presented in the general Theorem~\ref{thm:resmultiant} (and hence also in Theorem~\ref{thm:PerClassSingleAntenna} and the corollaries). The formal description of the optimal scheme in the upcoming subsection will be followed by a clarifying example in Section~\ref{subsec:example_scheme} that demonstrates the main idea behind the design.


\subsection{Description of the General Scheme}
The placement phase, which uses exactly the algorithm developed in \cite{MN14} for the case of $(\Lambda=K,M,N)$, is independent of $\mathcal{U},\Lc$, while the delivery phase is designed for any given $\Uc$, and will achieve the optimal worst-case delivery time stated in \eqref{eq:TS_L} and \eqref{eq:multi_delay}.
As mentioned, we will assume that any non zero $\mathcal{L}_{\lambda}$ satisfies $\mathcal{L}_{\lambda}\geq N_0, \forall \lambda\in[\Lambda]$.


\subsubsection{Cache Placement Phase \label{sec:SchemePlacement}}
The placement phase employs the original cache-placement algorithm of~\cite{MN14} corresponding to the scenario of having only $\Lambda$ users, each with their own cache. Hence --- recalling from~\cite{MN14} --- first each file $W^n$ is split into $\Lambda \choose \Lambda\gamma$ disjoint subfiles $W^n_\Tau$, for each $\Tau \subset [\Lambda]$, $|\Tau|=\Lambda\gamma$, and then each cache stores a fraction $\gamma$ of each file, as follows
\begin{equation}
\Zc_\lambda=\{W^n_\Tau :\Tau\ni\lambda,~ \forall n\in[N]\}.
\end{equation}

\subsubsection{Delivery Phase\label{sec:SchemeDelivery}}

For the purpose of the scheme description only, we will assume without loss of generality that $|\mathcal{U}_1| \geq |\mathcal{U}_2| \geq \dots \geq |\mathcal{U}_{\Lambda}|$ (any other case can be handled by simple relabeling of the caches), and we will use the notation $\mathcal{L}_\lambda \defeq |\Uc_\lambda|$. Furthermore, in a slight abuse of notation, we will consider here each $\mathcal{U}_\lambda$ to be an \emph{ordered vector} describing, in order, the users associated to cache $\lambda$.
We will also use
\begin{equation}\label{eq:Alambda}
\boldsymbol{s_{\lambda}} = (\mathcal{U}_\lambda \Vert \mathcal{U}_\lambda )_{N_0}, \lambda\in[\Lambda]
\end{equation}
to denote the $N_0$-fold concatenation of each $\mathcal{U}_\lambda$. Each such $N_0\mathcal{L}_{\lambda}$-length vector $\boldsymbol{s_{\lambda}}$ can be seen as the concatenation of $\mathcal{L}_\lambda$ different $N_0$-tuples $\boldsymbol{s_{\lambda,j}}$, $j=1,2,\dots, \mathcal{L}_\lambda$, i.e., each $\boldsymbol{s_{\lambda}}$ takes the form\footnote{Note also that having $\mathcal{L}_\lambda\geq N_0, \forall\lambda\in[\Lambda]$ guarantees that in any given $\boldsymbol{s_{\lambda,j}}, j\in[\mathcal{L}_\lambda]$, a user appears at most once.}
\begin{equation*}
\boldsymbol{s_{\lambda}} = \boldsymbol{s_{\lambda,1}} \Vert \boldsymbol{s_{\lambda,2}} \Vert \dots \Vert \underbrace{\boldsymbol{s_{\lambda,\mathcal{L}_\Lambda}}}_{N_0-\text{length}}.
\end{equation*}

The delivery phase commences with the demand vector $\dv$ being revealed to the server.
Delivery will consist of $\mathcal{L}_1$ rounds, where each round $j\in[\mathcal{L}_1]$ serves users
\begin{equation} \label{eq:UsersPerRound}
\mathcal{R}_j=\bigcup_{\lambda\in[\Lambda]} \big( \boldsymbol{s_{\lambda,j}}:\mathcal{L}_\lambda \geq j \big).
\end{equation} 

\paragraph*{Transmission scheme}
Once the demand vector $\dv$ is revealed to the transmitter, each requested subfile $W^{n}_{\Tau}$ (for any $n$ found in $\dv$) is further split into $N_0$ mini-files $\{W^{n}_{\Tau,l}\}_{l\in[N_0]}$. During round $j$, serving users in $\mathcal{R}_j$, we create $\Lambda \choose \Lambda\gamma+1$ sets $\mathcal{Q}\subseteq [\Lambda]$ of size $|\mathcal{Q}|=\Lambda\gamma+1$, and for each set $\mathcal{Q}$, we pick the set of users
\begin{equation}\label{eq:UsersServedPerXOR}
\chi_\mathcal{Q}=\bigcup_{\lambda\in \mathcal{Q}}\big( \boldsymbol{s_{\lambda,j}}:\mathcal{L}_\lambda \geq j \big).
\end{equation}
 If $\chi_\mathcal{Q} = \emptyset$, then there is no transmission, and we move to the next $\mathcal{Q}$. If $\chi_\mathcal{Q}\neq \emptyset$, the server --- during this round $j$ --- transmits the following vector\footnote{The transmitted-vector structure below draws from the structure in~\cite{LampirisEliaJsac18}, in the sense that it involves the linear combination of one or more \emph{Zero Forcing} precoded (ZF-precoded) vectors of subfiles that are labeled (as we see below) in the spirit of \cite{MN14}.}
 \begin{equation} \label{eq:TransmitSignalGeneral}
\xv_{\chi_{\mathcal{Q}}}=\!\!\!\!\sum_{\lambda\in \mathcal{Q}:\mathcal{L}_\lambda \geq j}\!\!\!\!\mathbf{H}^{-1}_{\boldsymbol{s_{\lambda,j}}}\cdot
		\begin{bmatrix}
    	 W^{d_{\boldsymbol{s_{\lambda,j}}(1)}}_{\mathcal{Q}\backslash{\{\lambda\}},l} &\dots &
   		W^{d_{\boldsymbol{s_{\lambda,j}}(N_0)}}_{\mathcal{Q}\backslash{\{\lambda\}},l}
		\end{bmatrix}^T
 \end{equation}
 where $W^{d_{\boldsymbol{s_{\lambda,j}}(k)}}_{\mathcal{Q}\backslash{\{\lambda\}},l}$ is a mini-file intended for user $\boldsymbol{s_{\lambda,j}}(k)$, i.e., for the user labelled by the $k$th entry of vector $\boldsymbol{s_{\lambda,j}}$ . The choice of $l$ is sequential, guaranteeing that no subfile $W^{d_{\boldsymbol{s_{\lambda,j}}(k)}}_{\mathcal{Q}\backslash{\{\lambda\}},l}$ is transmitted twice. Since each user appears in $\boldsymbol{s_{\lambda}}$ (and consequently in $\bigcup_{j\in[\mathcal{L}_1]} \mathcal{R}_j$) exactly $N_0$ times, at the end of the $\mathcal{L}_1$ rounds, all the $N_0$ mini-files $W^{d_{\boldsymbol{s_{\lambda,j}}(k)}}_{\mathcal{Q}\backslash{\{\lambda\}},l}$, $l\in[N_0]$ will be sent once. In the above, $\mathbf{H}^{-1}_{\boldsymbol{s_{\lambda,j}}}$ denotes the inverse of the channel matrix between the $N_0$ transmit antennas and the users in vector $\boldsymbol{s_{\lambda,j}}$.

\paragraph*{Decoding}
Directly from \eqref{eq:TransmitSignalGeneral}, we see that each receiver $\boldsymbol{s_{\lambda,j}}(k)$ obtains a received signal whose noiseless version takes the form
\[ y_{\boldsymbol{s_{\lambda,j}}(k)}=W^{d_{\boldsymbol{s_{\lambda,j}}(k)}}_{\mathcal{Q}\backslash{\{\lambda\},l}} + \iota_{\boldsymbol{s_{\lambda,j}}(k)} \]
where $\iota_{\boldsymbol{s_{\lambda,j}}(k)}$ is the $k$th entry of the interference vector
\begin{equation}
\label{eq:received48}
\sum_{\lambda'\in \mathcal{Q}\setminus{\{\lambda\}}:\mathcal{L}_{\lambda'}\geq j}\!\!\!\mathbf{H}^{-1}_{\boldsymbol{s_{\lambda',j}}}\cdot
		\begin{bmatrix}
    	 W^{d_{\boldsymbol{s_{\lambda',j}}(1)}}_{\mathcal{Q}\backslash{\{\lambda'\},l}} &\dots &
   		W^{d_{\boldsymbol{s_{\lambda',j}}(N_0)}}_{\mathcal{Q}\backslash{\{\lambda'\},l}}
		\end{bmatrix}^T .
\end{equation}
In the above, we see that the entire interference term $\iota_{\boldsymbol{s_{\lambda,j}}(k)}$ experienced by receiver $\boldsymbol{s_{\lambda,j}}(k)$, can be removed (cached-out) because all appearing subfiles $W^{d_{\boldsymbol{s_{\lambda',j}}(1)}}_{\mathcal{Q}\backslash{\{\lambda'\},l}}, \dots, W^{d_{\boldsymbol{s_{\lambda',j}}(N_0)}}_{\mathcal{Q}\backslash{\{\lambda'\},l}}$, for all $\lambda'\in \mathcal{Q}\setminus{\{\lambda\}}, \mathcal{L}_{\lambda'}\geq j$, can be found in cache $\lambda$ associated to this user, simply because $\lambda\in \mathcal{Q}\backslash\{\lambda'\}$.

This completes the proof of the scheme for the multi-antenna case.

%
%

\subsubsection{Small modification for the single antenna case\label{sec:schemeSingleAntenna}}For the single-antenna case, the only difference is that now $\boldsymbol{s_{\lambda}} = \mathcal{U}_\lambda$, and that each transmitted vector in~\eqref{eq:TransmitSignalGeneral} during round $j$, becomes a scalar of the form\footnote{A similar transmission method can be found also in the work of \cite{JinCaireGlobecom16} for the setting of decentralized coded caching with reduced subpacketization.}
 \begin{equation} \label{eq:TransmitSignalSingleAntenna}
x_{\chi_{\mathcal{Q}}}=\!\!\!\!\bigoplus_{\lambda\in \mathcal{Q}:\mathcal{L}_\lambda \geq j} W^{d_{\boldsymbol{s_{\lambda,j}}}}_{\mathcal{Q}\backslash{\{\lambda\}},1}.
 \end{equation}
 The rest of the details from the general scheme, as well as the subsequent calculation of the delay, follow directly.

\subsection{Calculation of Delay}

To first calculate the delay needed to serve the users in $\mathcal{R}_j$ during round $j$, we recall that there are $\Lambda \choose \Lambda\gamma+1$ sets $$\chi_\mathcal{Q}=\bigcup_{\lambda\in \mathcal{Q}}\big( \mathcal{U}_{\lambda}(j):\mathcal{L}_\lambda \geq j \big), \mathcal{Q}\subseteq [\Lambda]$$ of users, and we recall that $|\mathcal{U}_1| \geq |\mathcal{U}_2| \geq \dots \geq |\mathcal{U}_{\Lambda}|$. For each such non-empty set, there is a transmission. Furthermore we see that for $a_j\dfn \Lambda - \frac{|\mathcal{R}_j|}{N_0}$, there are ${a_j \choose \Lambda\gamma+1}$ such sets $\chi_\mathcal{Q}$ which are empty, which means that round $j$ consists of
\begin{equation}\label{eq:totsubround}
{\Lambda \choose \Lambda\gamma+1}-{a_j \choose \Lambda\gamma+1}
\end{equation}
transmissions.

Since each file is split into ${\Lambda\choose \Lambda\gamma}N_0$ subfiles, the duration of each such transmission is
\begin{equation}\label{eq:dureachtransmission}
\frac{1}{{\Lambda\choose \Lambda\gamma}N_0}
\end{equation}
and thus summing over all $\mathcal{L}_1$ rounds, the total delay takes the form
\begin{equation}\label{eq:totdelay1}
T=\frac{\sum_{j=1}^{\mathcal{L}_1}{{\Lambda \choose \Lambda\gamma+1}-{a_j \choose \Lambda\gamma+1}}}{{\Lambda\choose \Lambda\gamma}{N_0}}
\end{equation}
which, after some basic algebraic manipulation (see Appendix~\ref{sec:BinomialChangeProof} for the details), takes the final form
\begin{equation}\label{eq:totdelay2}
T=\frac{1}{N_0}\frac{\sum_{r=1}^{\Lambda-\Lambda\gamma}\mathcal{L}_r{\Lambda-r\choose \Lambda\gamma}}{{\Lambda\choose \Lambda\gamma}}
\end{equation}
which concludes the achievability part of the proof. \qed

\subsection{Scheme Example: $K=N=15$, $\Lambda=3$, $N_0=2$ and $\Lc=(8,5,2)$\label{subsec:example_scheme}}

Consider a scenario with $K=15$ users $\{1,2,\dots,15\}$, a server equipped with $N_0=2$ transmitting antennas that stores a library of $N=15$ equally-sized files $W^1,W^2,\dots,W^{15}$, and consider $\Lambda=3$ helper caches, each of size equal to $M=5$ units of file.

In the cache placement phase, we split each file $W^n$ into $3$ equally-sized disjoint subfiles denoted by $W^n_{1},W^n_{2},W^n_{3}$ and as in \cite{MN14}, each cache $\lambda$ stores $W^n_{\lambda}, \forall n\in [15]$.

We assume that in the subsequent cache assignment, users $\mathcal{U}_1=(1,2,3,4,5,6,7,8)$ are assigned to helper node $1$, users $\mathcal{U}_2=(9,10,11,12,13)$ to helper node $2$ and users $\mathcal{U}_3=(14,15)$ to helper node $3$. This corresponds to a profile $\Lc=(8,5,2)$.
We also assume without loss of generality that the demand vector is $\dv=(1,2,\dots,15)$.

Delivery takes place in $|\mathcal{U}_1|=8$ rounds, and each round will serve either $N_0=2$ users or no users from each of the following three ordered user groups
\begin{align*}
\boldsymbol{s_1}&=\mathcal{U}_1 || \mathcal{U}_1 = (1,2,\dots,7,8,1,2,\dots,7,8),\\
\boldsymbol{s_2}&=\mathcal{U}_2 || \mathcal{U}_2 = (9,10,11,12,13,9,10,11,12,13),\\
\boldsymbol{s_3}&=\mathcal{U}_3 || \mathcal{U}_3 = (14,15,14,15).
\end{align*}
Specifically, rounds 1 through 8, will respectively serve the following sets of users
\begin{align*}
\mathcal{R}_1&=\{1,2,9,10,14,15\}\\
\mathcal{R}_2&=\{3,4,11,12,14,15\}\\
\mathcal{R}_3&=\{5,6,13,9\}\\
\mathcal{R}_4&=\{7,8,10,11\}\\
\mathcal{R}_5&=\{1,2,12,13\}\\
\mathcal{R}_6&=\{3,4\}\\
\mathcal{R}_7&=\{5,6\}\\
\mathcal{R}_8&=\{7,8\}.
\end{align*}
Before transmission, each requested subfile $W^n_\Tau$ is further split into $N_0=2$ mini-files $W^n_{\Tau,1}$ and $W^n_{\Tau,2}$.  As noted in the general description of the scheme, the transmitted vector structure within each round, draws from \cite{LampirisEliaJsac18} as it employs the linear combination of ZF-precoded vectors. In the first round, the server transmits, one after the other, the following $3$ vectors
\begin{align}
\label{eq:ex1}
\xv_{\{1,2,9,10\}}=&\mathbf{H}^{-1}_{\{1,2\}}
\begin{bmatrix}
W^1_{2,1}\\
W^2_{2,1}
\end{bmatrix}
+
\mathbf{H}^{-1}_{\{9,10\}}
\begin{bmatrix}
W^{9}_{1,1}\\
W^{10}_{1,1}
\end{bmatrix}\\
\label{eq:ex2}
\xv_{\{1,2,14,15\}}=&\mathbf{H}^{-1}_{\{1,2\}}
\begin{bmatrix}
W^1_{3,1}\\
W^2_{3,1}
\end{bmatrix}
+
\mathbf{H}^{-1}_{\{14,15\}}
\begin{bmatrix}
W^{14}_{1,1}\\
W^{15}_{1,1}
\end{bmatrix}\\
\label{eq:ex3}
\xv_{\{9,10,14,15\}}=&\mathbf{H}^{-1}_{\{9,10\}}
\begin{bmatrix}
W^9_{3,1}\\
W^{10}_{3,1}
\end{bmatrix}
+
\mathbf{H}^{-1}_{\{14,15\}}
\begin{bmatrix}
W^{14}_{2,1}\\
W^{15}_{2,1}
\end{bmatrix}
\end{align}
where $\mathbf{H}^{-1}_{\{i,j\}}$ is the zero-forcing (ZF) precoder\footnote{Instead of ZF, one can naturally use a similar precoder with potentially better performance in different SNR ranges.} that inverts the channel $\mathbf{H}_{\{i,j\}}=[\mathbf{h}_i^T \mathbf{h}_j^T]$ from the transmitter to users $i$ and $j$.
To see how decoding takes place, let us first focus on users 1 and 2 during the transmission of $\xv_{\{1,2,9,10\}}$, where we see that, due to ZF precoding, the users' respective received signals take the form
\begin{align}
&y_1=W^1_{2,1}+\underbrace{\mathbf{h}_1^T\mathbf{H}^{-1}_{\{9,10\}}
\begin{bmatrix}
W^{9}_{1,1}\\
W^{10}_{1,1}
\end{bmatrix}}_{\text{interference}}+w_1\\
&y_2=W^2_{2,1}+\underbrace{\mathbf{h}_2^T\mathbf{H}^{-1}_{\{9,10\}}
\begin{bmatrix}
W^{9}_{1,1}\\
W^{10}_{1,1}
\end{bmatrix}}_{\text{interference}}+w_2.
\end{align}
Users 1 and 2 use their cached content in cache node 1, to remove files $W^9_{1,1},W^{10}_{1,1}$, and can thus directly decode their own desired subfiles.
The same procedure is applied to the remaining users served in the first round.

Similarly, in the second round, we have
 \begin{align}
\xv_{\{3,4,11,12\}}=&\mathbf{H}^{-1}_{\{3,4\}}
\begin{bmatrix}
W^3_{2,1}\\
W^4_{2,1}
\end{bmatrix}
+
\mathbf{H}^{-1}_{\{11,12\}}
\begin{bmatrix}
W^{11}_{1,1}\\
W^{12}_{1,1}
\end{bmatrix}\\
\xv_{\{3,4,14,15\}}=&\mathbf{H}^{-1}_{\{3,4\}}
\begin{bmatrix}
W^3_{3,1}\\
W^4_{3,1}
\end{bmatrix}
+
\mathbf{H}^{-1}_{\{14,15\}}
\begin{bmatrix}
W^{14}_{1,2}\\
W^{15}_{1,2}
\end{bmatrix}\\
\xv_{\{11,12,14,15\}}=&\mathbf{H}^{-1}_{\{11,12\}}
\begin{bmatrix}
W^{11}_{3,1}\\
W^{12}_{3,1}
\end{bmatrix}
+
\mathbf{H}^{-1}_{\{14,15\}}
\begin{bmatrix}
W^{14}_{2,2}\\
W^{15}_{2,2}
\end{bmatrix}
\end{align}
and again in each round, each pair of users can cache-out some of the files, and then decode their own file due to the ZF precoder.

The next three transmissions, corresponding to the third round, are as follows
\begin{align*}
&\xv_{\{5,6,13,9\}}=\mathbf{H}^{-1}_{\{5,6\}}
\begin{bmatrix}
W^5_{2,1}\\
W^6_{2,1}
\end{bmatrix}
+
\mathbf{H}^{-1}_{\{13,9\}}
\begin{bmatrix}
W^{13}_{1,1}\\
W^9_{1,2}
\end{bmatrix}
\\[3pt]
&\xv_{\{5,6\}}=\mathbf{H}^{-1}_{\{5,6\}}
\begin{bmatrix}
W^5_{3,1}\\
W^6_{3,1}
\end{bmatrix} \ \
\xv_{\{13,9\}}=\mathbf{H}^{-1}_{\{13,9\}}
\begin{bmatrix}
W^{13}_{3,1}\\
W^{9}_{3,2}
\end{bmatrix}
\end{align*}
where the transmitted vectors $\xv_{\{5,6\}}$ and $\xv_{\{13,9\}}$ simply use zero-forcing.
Similarly round 4 serves the users in $\mathcal{R}_4$ by sequentially sending
\begin{align}
&\xv_{\{7,8,10,11\}}=\mathbf{H}^{-1}_{\{7,8\}}
\begin{bmatrix}
W^7_{2,1}\\
W^8_{2,1}
\end{bmatrix}
+
\mathbf{H}^{-1}_{\{10,11\}}
\begin{bmatrix}
W^{10}_{1,2}\\
W^{11}_{1,2}
\end{bmatrix}
\\[3pt]
&\xv_{\{7,8\}}=\mathbf{H}^{-1}_{\{7,8\}}
\begin{bmatrix}
W^7_{3,1}\\
W^8_{3,1}
\end{bmatrix} \ \
\xv_{\{10,11\}}=\mathbf{H}^{-1}_{\{10,11\}}
\begin{bmatrix}
W^{10}_{3,2}\\
W^{11}_{3,2}
\end{bmatrix}
\end{align}
and round 5 serves the users in $\mathcal{R}_5$ by sequentially sending
\begin{align}
&\xv_{\{1,2,12,13\}}=\mathbf{H}^{-1}_{\{1,2\}}
\begin{bmatrix}
W^1_{2,2}\\
W^2_{2,2}
\end{bmatrix}
+
\mathbf{H}^{-1}_{\{12,13\}}
\begin{bmatrix}
W^{12}_{1,2}\\
W^{13}_{1,2}
\end{bmatrix}
\\[3pt]
&\xv_{\{1,2\}}=\mathbf{H}^{-1}_{\{1,2\}}
\begin{bmatrix}
W^1_{3,2}\\
W^2_{3,2}
\end{bmatrix} \ \
\xv_{\{12,13\}}=\mathbf{H}^{-1}_{\{12,13\}}
\begin{bmatrix}
W^{12}_{3,2}\\
W^{13}_{3,2}
\end{bmatrix}.
\end{align}
Finally, for the remaining rounds $6,7,8$ which respectively involve user sets $\mathcal{R}_6,\mathcal{R}_7$ and $\mathcal{R}_8$ that are connected to the same helper cache 1, data is delivered using the following standard ZF-precoded transmissions
\begin{align*}
&\xv_{\{3,4\}}=\mathbf{H}^{-1}_{\{3,4\}}
\begin{bmatrix}
W^3_{2,2}||W^3_{3,2}\\
W^4_{2,2}||W^4_{3,2}
\end{bmatrix}\\[3pt]
&\xv_{\{5,6\}}=\mathbf{H}^{-1}_{\{5,6\}}
\begin{bmatrix}
W^5_{2,2}||W^5_{3,2}\\
W^6_{2,2}||W^6_{3,2}
\end{bmatrix}\\
&\xv_{\{7,8\}}=\mathbf{H}^{-1}_{\{7,8\}}
\begin{bmatrix}
W^{7}_{2,2}||W^{7}_{3,2}\\
W^{8}_{2,2}||W^{8}_{3,2}
\end{bmatrix}.
\end{align*}

The overall delivery time required to serve all users is \[T=\frac{1}{6}\cdot 15+\frac{1}{3}\cdot 3 =\frac{21}{6}\] where the first summand is for rounds 1 through 5, and the second summand is for rounds 6 through 8.

It is very easy to see that this delay remains the same --- given again worst-case demand vectors --- for any user-to-cache association $\mathcal{U}$ with the same profile $\Lc=(8,5,2)$. Every time, this delay matches the converse
\begin{align}
T^*((8,5,2))&\geq \frac{\sum_{r=1}^{2}\mathcal{L}_r{{3-r}\choose 1}}{2{3\choose 1}}= \frac{8\cdot 2+5\cdot 1}{6}=\frac{21}{6}
\end{align} of Theorem \ref{thm:resmultiant}.



\section{Information Theoretic Converse\label{sec:converse}}

Toward proving Theorems~\ref{thm:PerClassSingleAntenna} and~\ref{thm:resmultiant}, we develop a lower bound on the normalized delivery time in (\ref{eq:T*_def}) for each given user-to-cache association profile $\Lc$.  The proof technique is based on the breakthrough in~\cite{WanTP15} which --- for the case of $\Lambda = K$, where each user has their own cache --- employed index coding to bound the performance of coded caching. Part of the challenge here will be to account for having shared caches, and mainly to adapt the index coding approach to reflect non-uniform user-to-cache association classes.

We will begin with lower bounding the normalized delivery time $T(\mathcal{U},\dv,\chi)$, for any user-to-cache association $\mathcal{U}$, demand vector $\dv$ and a generic caching-delivery strategy $\chi$.

\paragraph*{Identifying the distinct problems} 
The caching problem is defined when the user-to-cache association $\mathcal{U}=\{\mathcal{U}_\lambda \}_{\lambda=1}^\Lambda$ and demand vector $\dv$ are revealed.
What we can easily see is that there are many combinations of $\{\mathcal{U}_\lambda \}_{\lambda=1}^\Lambda$ and $\dv$ that jointly result in the same coded caching problem. After all, any permutation of the file indices requested by users assigned to the same cache, will effectively result in the same coded caching problem.
As one can see, every \emph{distinct} coded caching problem is fully defined by $\{\dvlambda\}_{\lambda=1}^\Lambda$, where $\dvlambda$ denotes the vector of file indices requested by the users in $\mathcal{U}_\lambda$, i.e., requested by the $|\mathcal{U}_\lambda|$ users associated to cache $\lambda$. The analysis is facilitated by reordering the demand vector $\dv$ to take the form
\begin{equation}
\label{eq:OrderDemand2}
\dv(\Uc)\dfn (\boldsymbol{d_1}, \dots, \boldsymbol{d_\lambda}).\end{equation}
Based on this, we define the set of worst-case demands associated to a given profile $\Lc$, to be
$$\mathcal{D}_{\Lc} = \{\dv(\mathcal{U}): \dv\in \mathcal{D}_{wc}, \mathcal{U} \in \mathcal{U}_{\Lc}
\}$$ where $\mathcal{D}_{wc}$ is the set of demand vectors $\dv$ whose $K$ entries are all different (i.e., where $d_i \neq d_j, ~i,j\in[\Lambda],~i\neq j$, corresponding to the case where all users request different files). We will convert each such coded caching problem into an index coding problem.

\paragraph*{The corresponding index coding problem}
To make the transition to the index coding problem, each requested file $W^{\dvlambda(j)}$ is split into $2^\Lambda$ disjoint subfiles $W^{\dvlambda(j)}_\Tau,\Tau\in 2^{[\Lambda]}$ where $\Tau\subset[\Lambda]$ indicates the set of helper nodes in which $W^{\dvlambda(j)}_\Tau$ is cached\footnote{Notice that by considering a subpacketization based on the power set $2^{[\Lambda]}$, and by allowing for any possible size of these subfiles, the generality of the result is preserved. Naturally, this does not impose any sub-packetization related performance issues because this is done only for the purpose of creating a converse.}. Then --- in the context of index coding --- each subfile $W^{\dvlambda(j)}_\Tau$ can be seen as being requested by a different user that has as side information all the content $\Zc_\lambda$ of the same helper node $\lambda$. Naturally, no subfile of the form $W^{\dvlambda(j)}_\Tau, \; ~\Tau \ni\lambda$ is requested, because helper node $\lambda$ already has this subfile. Therefore the corresponding index coding problem is defined by $K2^{\Lambda-1}$ requested subfiles, and it is fully represented by the side-information graph $\mathcal{G}=(\mathcal{V}_{\mathcal{G}},\mathcal{E}_{\mathcal{G}})$, where $\mathcal{V}_{\mathcal{G}}$ is the set of vertices (each vertex/node representing a different subfile $W^{\dvlambda(j)}_\Tau, \Tau\not\ni\lambda$) and $\mathcal{E}_{\mathcal{G}}$ is the set of direct edges of the graph. Following standard practice in index coding, a directed edge from node $W^{\dvlambda(j)}_\Tau$ to $W^{\boldsymbol{d_{\lambda'}}(j')}_{\Tau'}$ exists if and only if $\lambda'\in\Tau$.  For any given $\mathcal{U}$, $\dv$ (and of course, for any scheme $\chi$) the total delay $T$ required for this index coding problem, is the completion time for the corresponding coded caching problem.

\paragraph*{Lower bounding $T(\Uc,\dv,\chi)$}
We are interested in lower bounding $T(\Uc,\dv,\chi)$ which represents the total delay required to serve the users for the index coding problem corresponding to the side-information graph $\mathcal{G}_{\Uc,\dv}$ defined by $\Uc,\dv,\chi$ or equivalently by $\dv(\Uc),\chi$.

In the next lemma, we remind the reader --- in the context of our setting --- the useful index-coding converse from~\cite{li2017cooperative}.

\begin{lemma}(Cut-set-type converse \cite{li2017cooperative})\label{cor_dof}
For a given $\Uc,\dv,\chi$, in the corresponding side information graph $\mathcal{G}_{\Uc,\dv}=(\mathcal{V}_{\mathcal{G}},\mathcal{E}_{\mathcal{G}})$ of the $N_0$-antenna MISO broadcast channel with $\mathcal{V}_{\mathcal{G}}$ vertices/nodes and $\mathcal{E}_{\mathcal{G}}$ edges, the following inequality holds
\begin{equation}\label{eq:indexbound}T\geq \frac{1}{N_0}\sum_{\smallV \in \mathcal{V_{J}}}|\smallV|
\end{equation}
for every acyclic induced subgraph $\mathcal{J}$ of $\mathcal{G}_{\Uc,\dv}$, where $\mathcal{V}_{\mathcal{J}}$ denotes the set of nodes of the subgraph $\mathcal{J}$, and where $|\smallV|$ is the size of the message/subfile/node $\smallV$.
\end{lemma}

\vspace{3pt}\emph{Proof.} The above lemma draws from~\cite[Corollary 1]{li2017cooperative} (see also~\cite[Corollary 2]{Sadeghi:16} for a simplified version), and is easily proved in the Appendix Section~\ref{proof:cor_dof}.\vspace{3pt}

\paragraph*{Creating large acyclic subgraphs}
Lemma~\ref{cor_dof} suggests the need to create (preferably large) acyclic subgraphs of $\mathcal{G}_{\mathcal{U},\dv}$. The following lemma describes how to properly choose a set of nodes to form a large acyclic subgraph.

\begin{lemma}\label{lem:cons_acyclic}
An acyclic subgraph $\mathcal{J}$ of $\mathcal{G}_{\Uc,\dv}$ corresponding to the index coding problem defined by $\Uc,\dv,\chi$ for any $\Uc$ with profile $\Lc$, is designed here to consist of all subfiles $W^{\boldsymbol{d_{\sigma_{s}(\lambda)}}(j)}_{\Tau_{\lambda}},~\forall j\in [\mathcal{L}_{\lambda}],~\forall \lambda\in [\Lambda]$ for all $\Tau_{\lambda}\subseteq [\Lambda]\setminus \{\sigma_s(1),\dots,\sigma_s(\lambda)\}$ where $\sigma_s\in S_{\Lambda}$ is the permutation such that $|\mathcal{U}_{\sigma_s(1)}|\geq |\mathcal{U}_{\sigma_s(2)}|\geq\dots\geq |\mathcal{U}_{\sigma_s(\Lambda)}|$.
\end{lemma}
\vspace{3pt}\emph{Proof.} The proof, which can be found in the Appendix Section~\ref{proof:cons_acyclic}, is an adaptation of~\cite[Lemma 1]{WanTP15} to the current setting.\vspace{3pt}

\begin{remark}
The choice of the permutation $\sigma_s$ is critical for the development of a tight converse. Any other choice $\sigma\in S_\Lambda$ may result --- in some crucial cases --- in an acyclic subgraph with a smaller number of nodes and therefore a looser bound. This approach here deviates from the original approach in \cite[Lemma 1]{WanTP15}, which instead considered --- for each $\dv,\chi$, for the uniform user-to-cache association case of $K = \Lambda$ --- the set of \emph{all} possible permutations, that jointly resulted in a certain symmetry that is crucial to that proof. Here in our case, such symmetry would not serve the same purpose as it would dilute the non-uniformity in $\Lc$ that we are trying to capture. Our choice of a single carefully chosen permutation, allows for a bound which --- as it turns out --- is tight even in non-uniform cases.
The reader is also referred to Section~\ref{subsec:example} for an explanatory example.
\end{remark}
Having chosen an acyclic subgraph according to Lemma~\ref{lem:cons_acyclic}, we return to Lemma~\ref{cor_dof} and form --- by adding the sizes of all subfiles associated to the chosen acyclic graph --- the following lower bound
\begin{equation}
T(\Uc,\dv,\chi)\geq T^{LB}(\Uc,\dv,\chi)
\end{equation} where
\begin{align}
&T^{LB}(\Uc,\dv,\chi) \defeq  \frac{1}{N_0}\Bigg( \sum_{j=1}^{\mathcal{L}_{1}}\sum_{\Tau_{1}\subseteq [\Lambda]\setminus \{\sigma_s(1)\}}|W^{\boldsymbol{d_{\sigma_s(1)}}(j)}_{\Tau_{1}}|\nonumber \\
&+ \sum_{j=1}^{\mathcal{L}_{2}}\sum_{\Tau_{2}\subseteq [\Lambda]\setminus \{\sigma_s(1),\sigma_s(2)\}}|W^{\boldsymbol{d_{\sigma_s(2)}}(j)}_{\Tau_{2}}|+\dots \nonumber \\
&+ \sum_{j=1}^{\mathcal{L}_{\Lambda}}\sum_{\Tau_{\Lambda}\subseteq [\Lambda]\setminus \{\sigma_s(1),\dots,\sigma_s(\Lambda)\}}|W^{\boldsymbol{d_{\sigma_s(\Lambda)}}(j)}_{\Tau_{\Lambda}}|
\Bigg). \label{eq:TLB}
\end{align}

Our interest lies in a lower bound for the worst-case delivery time/delay associated to profile $\Lc$. Such a worst-case naturally corresponds to the scenario where all users request different files, i.e., where all the entries of the demand vector $\dv(\mathcal{U})$ are different. The corresponding lower bound can be developed by averaging over worst-case demands. Recalling our set $\mathcal{D}_{\Lc}$, the worst-case delivery time can thus be written as
\begin{align}
T^*(\Lc)&\defeq \min_{\chi} \max_{(\mathcal{U},\dv) \in (\mathcal{U}_{\Lc},[N]^K)} T(\mathcal{U},\dv,\chi)\\
&\overset{(a)}{\geq} \min_{\chi} \frac{1}{|\mathcal{D}_{\Lc}|} \sum_{\dv(\mathcal{U}) \in \mathcal{D}_{\Lc}} T(\dv(\mathcal{U}),\chi)\label{eq:alternativedefinitionofT}
\end{align} where in step (a), we used the following change of notation $T(\dv(\mathcal{U}),\chi)\dfn T(\mathcal{U},\dv,\chi)$ and averaged over worst-case demands.

With a given class/profile $\Lc$ in mind, in order to construct $\mathcal{D}_{\Lc}$ (so that we can then average over it), we will consider all demand vectors $\dv\in \mathcal{D}_{wc}$ for all permutations $\pi\in S_{\Lambda}$. 
Then for each $\dv$, we create the following set of $\Lambda$ vectors
\begin{align*}
&\boldsymbol{d^{'}_1}= (d_1 : d_{\mathcal{L}_1}),\\
&\boldsymbol{d^{'}_2}= (d_{\mathcal{L}_1+1} : d_{\mathcal{L}_1+\mathcal{L}_2}),\\
& \vdots \\
&\boldsymbol{d^{'}_{\Lambda}}= (d_{\sum_{i=1}^{\Lambda-1}\mathcal{L}_{i}~+1} : d_{K})
\end{align*}
and for each permutation $\pi\in S_{\Lambda}$ applied to the set $\{1,2,\dots,\Lambda\}$, a demand vector $\dv(\mathcal{U})$ is constructed as follows
\begin{align}
\dv(\mathcal{U})&\dfn(\boldsymbol{d_1},\boldsymbol{d_2},\dots,\boldsymbol{d_\Lambda})\\
&=(\boldsymbol{d^{'}_{\pi^{-1}(1)}},\boldsymbol{d^{'}_{\pi^{-1}(2)}},\dots,\boldsymbol{d^{'}_{\pi^{-1}(\Lambda)}}).
\end{align}
This procedure is repeated for all $\Lambda!$ permutations ${\pi\in S_{\Lambda}}$ and all $P(N,K)$ worst-case demands $\dv\in \mathcal{D}_{wc}$. This implies that the cardinality of $\mathcal{D}_{\Lc}$ is ${|\mathcal{D}_{\Lc}|=P(N,K)\cdot \Lambda!}$.

Using this designed set $\mathcal{D}_{\Lc}$, now the optimal worst-case delivery time in (\ref{eq:alternativedefinitionofT}) is bounded as
\begin{align}
T^{*}(\Lc) 
&= \min_{\chi}T(\Lc,\chi)\\
& \geq  \min_{\chi} \frac{1}{P(N,K)\Lambda!} \sum_{\dv(\mathcal{U}) \in \mathcal{D}_{\Lc}} T^{LB}(\dv(\mathcal{U}),\chi) \label{eq:lowerboundcompact}
\end{align}
where $T^{LB}(\dv(\mathcal{U}),\chi)$ is given by (\ref{eq:TLB}) for each reordered demand vector $\dv(\mathcal{U})\in \mathcal{D}_{\Lc}$. Rewriting the summation in (\ref{eq:lowerboundcompact}), we get
\begin{align}\label{eq:longinequality}
&\sum_{\dv(\mathcal{U})\in \mathcal{D}_{\Lc}}  T^{LB}(\dv(\mathcal{U}),\chi)= \nonumber \\
&\frac{1}{N_0}\sum_{i=0}^{\Lambda}\sum_{n\in[N]}\sum_{\Tau\subseteq[\Lambda]:|\Tau|=i} |W^n_{\Tau}| \cdot \underbrace{\sum_{\dv(\mathcal{U})\in \mathcal{D}_{\Lc}} \mathds{1}_{\mathcal{V}_{\mathcal{J}_s^{\dv(\mathcal{U})}}}(W^n_{\Tau})}_{\defeq Q_{i}(W^n_\Tau)}
\end{align} where $\mathcal{V}_{\mathcal{J}_s^{\dv(\mathcal{U})}}$ is the set of vertices in the acyclic subgraph chosen according to Lemma \ref{lem:cons_acyclic} for a given $\dv(\mathcal{U})$. In the above, $\mathds{1}_{\mathcal{V}_{\mathcal{J}_s^{\dv(\mathcal{U})}}}(W^n_{\Tau})$ denotes the indicator function which takes the value of 1 only if $W^n_{\Tau} \subset \mathcal{V}_{\mathcal{J}_s^{\dv(\mathcal{U})}}$, else it is set to zero.

A crucial step toward removing the dependence on $\Tau$, comes from the fact that
\begin{align}\label{eq:Qi}
Q_{i} &= Q_{i}(W^n_\Tau)\dfn \sum_{\dv(\mathcal{U})\in \mathcal{D}_{\Lc}} \mathds{1}_{\mathcal{V}_{\mathcal{J}_s^{\dv(\Uc)}}}(W^n_{\Tau}) \nonumber\\
=&{N-1 \choose K-1}\sum_{r=1}^{\Lambda}P(\Lambda-i-1,r-1)(\Lambda-r)!\mathcal{L}_{r} \nonumber\\
&\times P(K-1,\mathcal{L}_{r}-1) (K-\mathcal{L}_{r})! (\Lambda-i)
\end{align}
where we can see that the total number of times a specific subfile appears --- in the summation in \eqref{eq:longinequality}, over the set of all possible $\dv(\mathcal{U})  \in \mathcal{D}_{\Lc}$, and given our chosen permutation $\sigma_s$ 
--- is not dependent on the subfile itself but is dependent only on the number of caches $i=|\Tau|$ storing that subfile. The proof of \eqref{eq:Qi} can be found in Section~\ref{proof:lemmaQi}.

In the spirit of~\cite{WanTP15}, defining
\begin{equation}
x_i\dfn\sum_{n\in[N]}\sum_{\Tau\subseteq[\Lambda]:|\Tau|=i}|W^n_{\Tau}|
\end{equation}
to be the total amount of data stored in exactly $i$ helper nodes, we see that
\begin{equation}\label{eq:sumfiles}
N=\sum_{i=0}^{\Lambda}x_i=\sum_{i=0}^{\Lambda}\sum_{n\in[N]}\sum_{\Tau\subseteq[\Lambda]:|\Tau|=i}|W^n_{\Tau}|
\end{equation}
and we see that combining \eqref{eq:lowerboundcompact}, \eqref{eq:longinequality} and \eqref{eq:Qi}, gives
\begin{equation}\label{eq:compacteq}
T(\Lc,\chi)\geq \frac{1}{N_{0}}\sum_{i=0}^{\Lambda}\frac{Q_{i}}{P(N,K)\Lambda!}x_{i}.
\end{equation}

Now substituting \eqref{eq:Qi} into \eqref{eq:compacteq}, after some algebraic manipulations, we get that
\begin{align}
T(\Lc,\chi)&\geq \frac{1}{N_0}\sum_{i=0}^{\Lambda}\frac{\sum_{r=1}^{\Lambda-i}\mathcal{L}_{r} {\Lambda-r\choose i}}{N{\Lambda\choose i}}x_{i} \label{eq:LBwithxi}\\
&=\frac{1}{N_0}\sum_{i=0}^{\Lambda}\frac{x_{i}}{N}c_{i} \label{eq:LBwithxi_2}
\end{align}
where $c_{i}\triangleq \frac{\sum_{r=1}^{\Lambda-i}\mathcal{L}_r{\Lambda-r\choose i}}{{\Lambda\choose i}}$ decreases with $i\in \{0,1,\dots,\Lambda\}$.
The proof of the transition from \eqref{eq:compacteq} to \eqref{eq:LBwithxi}, as well as the monotonicity proof for the sequence $\{c_i\}_{i\in [\Lambda]\cup \{0\}}$, are given in Appendix Sections~\ref{proof:transition} and \ref{sec:monotonicity} respectively.

Under the file-size constraint given in \eqref{eq:sumfiles}, and given the following cache-size constraint
\begin{equation}
\sum_{i=0}^{\Lambda}i \cdot x_{i}\leq  \Lambda M \label{eq:constr2}
\end{equation}
the expression in~\eqref{eq:LBwithxi} serves as a lower bound on the delay of any caching-and-delivery scheme $\chi$ whose caching policy implies a set of $\{x_i\}$.

We then employ the Jensen's-inequality based technique of \cite[Proof of Lemma 2]{YuMA16} to minimize the expression in \eqref{eq:LBwithxi}, over all admissible $\{x_i\}$. Hence for any integer $\Lambda\gamma$, we have
\begin{equation}\label{eq:optimization1}
T(\Lc,\chi)\geq \frac{\sum_{r=1}^{\Lambda-\Lambda\gamma}\mathcal{L}_r{\Lambda-r\choose \Lambda\gamma}}{{\Lambda\choose \Lambda\gamma}}
\end{equation}
whereas for all other values of $\Lambda\gamma$, this is extended to its convex lower envelop.
The detailed derivation of \eqref{eq:optimization1} can again be found in Appendix Section~\ref{lastproof}.

This concludes lower bounding  $\max_{(\mathcal{U},\dv) \in (\mathcal{U}_{\Lc},[N]^K)} T(\mathcal{U},\dv,\chi)$, and thus --- given that the right hand side of \eqref{eq:optimization1} is independent of $\chi$ --- lower bounds the performance for any scheme $\chi$, which hence concludes the proof of the converse for Theorem~\ref{thm:resmultiant} (and consequently for Theorem~\ref{thm:PerClassSingleAntenna} after setting $N_0 = 1$).  \qed

\subsection{Proof of the Converse for Corollary~\ref{cor:ressymMulti} \label{sec:ConverseUniform}}
For the uniform case of $\Lc = [\frac{K}{\Lambda},\frac{K}{\Lambda},\dots,\frac{K}{\Lambda}]$, the lower bound in \eqref{eq:optimization1} becomes
\begin{align}  \frac{1}{N_0}\frac{\sum_{r=1}^{\Lambda-\Lambda\gamma}\mathcal{L}_r{\Lambda-r\choose \Lambda\gamma}}{{\Lambda\choose \Lambda\gamma}}
& =\frac{1}{N_0}\frac{K}{\Lambda}\frac{\sum_{r=1}^{\Lambda-\Lambda\gamma}{\Lambda-r\choose \Lambda\gamma}}{{\Lambda\choose \Lambda\gamma}} \\
  &\overset{(a)}{=} \frac{1}{N_0}\frac{K}{\Lambda}\frac{{\Lambda \choose \Lambda\gamma+1}}{{\Lambda\choose \Lambda\gamma}}  \\
	& = \frac{K(1-\gamma)}{N_0(\Lambda\gamma+1)}
  \end{align}
where the equality in step (a) is due to Pascal's triangle.  \qed



\subsection{Example for $N=K=9$, $N_{0}=2$ and $\Lc=(4,3,2)$\label{subsec:example}}
We here give an example of deriving the converse for Theorem \ref{thm:resmultiant}, emphasizing on how to convert the caching problem to the index-coding problem, and how to choose acyclic subgraphs. We consider the case of having $K=9$ receiving users, and a transmitter with $N_{0}=2$ transmit antennas having access to a library of $N=9$ files of unit size. We also assume that there are $\Lambda=3$ caching nodes, of average normalized cache capacity $\gamma$. We will focus on deriving the bound for user-to-cache association profile $\Lc=(4,3,2)$, meaning that we are interested in the setting where one cache is associated to $4$ users, one cache to $3$ users and one cache associated to $2$ users.

Each file $W^n$ is split into $2^{\Lambda}=8$ disjoint subfiles $W^{i}_{\Tau}, \Tau\in 2^{[3]}$ where each $\Tau$ describes the set of helper nodes in which $W^{i}_{\Tau}$ is cached. For instance, $W^1_{13}$ refers to the part of file $W^1$ that is stored in the first and third caching nodes.

As a first step, we present the construction of the set $\mathcal{D}_{\Lc}$. To this end, let us start by considering the demand $\dv=(1,2,3,4,5,6,7,8,9)$ and one of the $6$ permutations $\pi\in S_{3}$; for example, let us start by considering $\pi(1)=2,\pi(2)=3,\pi(3)=1$. Toward reordering $\dv$ to reflect $\Lc$, we construct
\begin{align*}
&\boldsymbol{d^{'}_1}=(1,2,3,4),
~~\boldsymbol{d^{'}_2}=(5,6,7),
~~\boldsymbol{d^{'}_3}=(8,9)
\end{align*}
to obtain the reordered demand vector
\begin{align*}
\dv(\mathcal{U})&=(\boldsymbol{d^{'}_{\pi^{-1}(1)}},\boldsymbol{d^{'}_{\pi^{-1}(2)}},\boldsymbol{d^{'}_{\pi^{-1}(3)}})\\
&=(\boldsymbol{d^{'}_{3}},\boldsymbol{d^{'}_{1}},\boldsymbol{d^{'}_{2}})
\end{align*}
which in turn yields $\boldsymbol{d_1}=(8,9),\boldsymbol{d_2}=(1,2,3,4),\boldsymbol{d_3}=(5,6,7)$. Similarly, we can construct the remaining $5$ demands $\dv(\mathcal{U})$ associated to the other $5$ permutations $\pi\in S_{3}$. Finally, the procedure is repeated for all other worst-case demand vectors. These vectors are part of set $\mathcal{D}_{\Lc}$.

With the users demands $\dv(\mathcal{U})$ known to the server, the delivery problem is translated into an index coding problem with a side information graph of $K 2^{\Lambda-1}=9\cdot 2^{2}$ nodes. For each requested file $W^{\boldsymbol{d_\lambda}(j)}$, we write down the $4$ subfiles that the requesting user does not have in its assigned cache. Hence, a given user of the caching problem requiring $4$ subfiles from the main server, is replaced by $4$ different new users in the index coding problem. Each of these users request a different subfile and are connected to the same cache $\lambda$ as the original user.
\begin{figure*}
\[
\begin{array}{c@{}c@{}ccc}
\boldsymbol{d_1}=(1,2,3,4),\boldsymbol{d_2}=(5,6,7),
&~~&
\boldsymbol{d_1}=(1,2,3,4),\boldsymbol{d_2}=(8,9),
&
\boldsymbol{d_1}=(5,6,7),\boldsymbol{d_2}=(1,2,3,4),\\

\boldsymbol{d_3}=(8,9)&~~&\boldsymbol{d_3}=(5,6,7)&\boldsymbol{d_3}=(8,9)\\
  ~~&~~&~~\\
\begin{array}{cccc}
  \underline{W^{1}_{\emptyset}} & \underline{W^{1}_{2}} & \underline{W^{1}_{3}} & \underline{W^{1}_{23}}\\
  \underline{W^{2}_{\emptyset}} & \underline{W^{2}_{2}} & \underline{W^{2}_{3}} & \underline{W^{2}_{23}}\\
  \underline{W^{3}_{\emptyset}} & \underline{W^{3}_{2}} & \underline{W^{3}_{3}} & \underline{W^{3}_{23}}\\
  \underline{W^{4}_{\emptyset}} & \underline{W^{4}_{2}} & \underline{W^{4}_{3}} & \underline{W^{4}_{23}}\\
  \underline{W^{5}_{\emptyset}} & W^{5}_{1} & \underline{W^{5}_{3}} & W^{5}_{13}\\
  \underline{W^{6}_{\emptyset}} & W^{6}_{1} & \underline{W^{6}_{3}} & W^{6}_{13}\\
  \underline{W^{7}_{\emptyset}} & W^{7}_{1} & \underline{W^{7}_{3}} & W^{7}_{13}\\
  \underline{W^{8}_{\emptyset}} & W^{8}_{1} & W^{8}_{2} & W^{8}_{12}\\
  \underline{W^{9}_{\emptyset}} & W^{9}_{1} & W^{9}_{2} & W^{9}_{12}\\
  \end{array} &~~&
\begin{array}{cccc}
  \underline{W^{1}_{\emptyset}} & \underline{W^{1}_{2}} & \underline{W^{1}_{3}} & \underline{W^{1}_{23}}\\
  \underline{W^{2}_{\emptyset}} & \underline{W^{2}_{2}} & \underline{W^{2}_{3}} & \underline{W^{2}_{23}}\\
  \underline{W^{3}_{\emptyset}} & \underline{W^{3}_{2}} & \underline{W^{3}_{3}} & \underline{W^{3}_{23}}\\
  \underline{W^{4}_{\emptyset}} & \underline{W^{4}_{2}} & \underline{W^{4}_{3}} & \underline{W^{4}_{23}}\\
  \underline{W^{5}_{\emptyset}} & W^{5}_{1} & \underline{W^{5}_{2}} & W^{5}_{12}\\
  \underline{W^{6}_{\emptyset}} & W^{6}_{1} & \underline{W^{6}_{2}} & W^{6}_{12}\\
  \underline{W^{7}_{\emptyset}} & W^{7}_{1} & \underline{W^{7}_{2}} & W^{7}_{12}\\
  \underline{W^{8}_{\emptyset}} & W^{8}_{1} & W^{8}_{3} & W^{8}_{13}\\
  \underline{W^{9}_{\emptyset}} & W^{9}_{1} & W^{9}_{3} & W^{9}_{13}\\

  \end{array} &
\begin{array}{cccc}
  \underline{W^{1}_{\emptyset}} & \underline{W^{1}_{1}} & \underline{W^{1}_{3}} & \underline{W^{1}_{13}} \\
  \underline{W^{2}_{\emptyset}} & \underline{W^{2}_{1}} & \underline{W^{1}_{3}} & \underline{W^{2}_{13}} \\
  \underline{W^{3}_{\emptyset}} & \underline{W^{3}_{1}} & \underline{W^{3}_{3}} & \underline{W^{3}_{13}} \\
  \underline{W^{4}_{\emptyset}} & \underline{W^{4}_{1}} & \underline{W^{4}_{3}} & \underline{W^{4}_{13}} \\
  \underline{W^{5}_{\emptyset}} & W^{5}_{2} & \underline{W^{5}_{3}} & W^{5}_{23} \\
  \underline{W^{6}_{\emptyset}} & W^{6}_{2} & \underline{W^{6}_{3}} & W^{6}_{23} \\
  \underline{W^{7}_{\emptyset}} & W^{7}_{2} & \underline{W^{7}_{3}} & W^{7}_{23} \\
  \underline{W^{8}_{\emptyset}} & W^{8}_{1} & W^{8}_{2} & W^{8}_{12} \\
  \underline{W^{9}_{\emptyset}} & W^{9}_{1} & W^{9}_{2} & W^{9}_{12} \\
  \end{array} \\
  ~~&~~&~~\\
\boldsymbol{d_1}=(5,6,7),\boldsymbol{d_2}=(8,9),
&~~&
\boldsymbol{d_1}=(8,9),\boldsymbol{d_2}=(1,2,3,4),
&
\boldsymbol{d_1}=(8,9),\boldsymbol{d_2}=(5,6,7),\\

\boldsymbol{d_3}=(1,2,3,4)&~~&\boldsymbol{d_3}=(5,6,7)&\boldsymbol{d_3}=(1,2,3,4)\\
  ~~&~~&~~\\
\begin{array}{cccc}
  \underline{W^{1}_{\emptyset}} & \underline{W^{1}_{1}} & \underline{W^{1}_{2}} & \underline{W^{1}_{12}} \\
  \underline{W^{2}_{\emptyset}} & \underline{W^{2}_{1}} & \underline{W^{2}_{2}} & \underline{W^{2}_{12}} \\
  \underline{W^{3}_{\emptyset}} & \underline{W^{3}_{1}} & \underline{W^{3}_{2}} & \underline{W^{3}_{12}} \\
  \underline{W^{4}_{\emptyset}} & \underline{W^{4}_{1}} & \underline{W^{4}_{2}} & \underline{W^{4}_{12}} \\
  \underline{W^{5}_{\emptyset}} & \underline{W^{5}_{2}} & W^{5}_{3} & W^{5}_{23} \\
  \underline{W^{6}_{\emptyset}} & \underline{W^{6}_{2}} & W^{6}_{3} & W^{6}_{23} \\
  \underline{W^{7}_{\emptyset}} & \underline{W^{7}_{2}} & W^{7}_{3} & W^{7}_{23} \\
  \underline{W^{8}_{\emptyset}} & W^{8}_{1} & W^{8}_{3} & W^{8}_{13} \\
  \underline{W^{9}_{\emptyset}} & W^{9}_{1} & W^{9}_{3} & W^{9}_{13} \\
  \end{array} &~~&
\begin{array}{cccc}
  \underline{W^{1}_{\emptyset}} & \underline{W^{1}_{1}} & \underline{W^{1}_{3}} & \underline{W^{1}_{13}} \\
  \underline{W^{2}_{\emptyset}} & \underline{W^{2}_{1}} & \underline{W^{2}_{3}} & \underline{W^{2}_{13}} \\
  \underline{W^{3}_{\emptyset}} & \underline{W^{3}_{1}} & \underline{W^{3}_{3}} & \underline{W^{3}_{13}} \\
  \underline{W^{4}_{\emptyset}} & \underline{W^{4}_{1}} & \underline{W^{4}_{3}} & \underline{W^{4}_{13}} \\
  \underline{W^{5}_{\emptyset}} & \underline{W^{5}_{1}} & W^{5}_{2} & W^{5}_{12} \\
  \underline{W^{6}_{\emptyset}} & \underline{W^{6}_{1}} & W^{6}_{2} & W^{6}_{12} \\
  \underline{W^{7}_{\emptyset}} & \underline{W^{7}_{1}} & W^{7}_{2} & W^{7}_{12} \\
  \underline{W^{8}_{\emptyset}} & W^{8}_{2} & W^{8}_{3} & W^{8}_{23} \\
  \underline{W^{9}_{\emptyset}} & W^{9}_{2} & W^{9}_{3} & W^{9}_{23} \\
  \end{array} &
  \begin{array}{cccc}
  \underline{W^{1}_{\emptyset}} & \underline{W^{1}_{1}} & \underline{W^{1}_{2}} & \underline{W^{1}_{12}} \\
  \underline{W^{2}_{\emptyset}} & \underline{W^{2}_{1}} & \underline{W^{2}_{2}} & \underline{W^{2}_{12}} \\
  \underline{W^{3}_{\emptyset}} & \underline{W^{3}_{1}} & \underline{W^{3}_{2}} & \underline{W^{3}_{12}} \\
  \underline{W^{4}_{\emptyset}} & \underline{W^{4}_{1}} & \underline{W^{4}_{2}} & \underline{W^{4}_{12}} \\
  \underline{W^{5}_{\emptyset}} & \underline{W^{5}_{1}} & W^{5}_{3} & W^{5}_{13} \\
  \underline{W^{6}_{\emptyset}} & \underline{W^{6}_{1}} & W^{6}_{3} & W^{6}_{13} \\
  \underline{W^{7}_{\emptyset}} & \underline{W^{7}_{1}} & W^{7}_{3} & W^{7}_{13} \\
  \underline{W^{8}_{\emptyset}} & W^{8}_{2} & W^{8}_{3} & W^{8}_{23} \\
  \underline{W^{9}_{\emptyset}} & W^{9}_{2} & W^{9}_{3} & W^{9}_{23} \
  \end{array} \\
\end{array}
\]
\caption{Nodes of the side information graphs corresponding to demand vector $\dv=(1,2,3,4,5,6,7,8,9)$ (profile $\Lc=(4,3,2)$).}
 \label{fig:graphs}
\end{figure*}
The nodes of the $6$ side-information graphs corresponding to the aforementioned vectors $\dv(\mathcal{U})$ (one for each permutation $\pi\in S_{3}$) for demand $\dv=(1,2,3,4,5,6,7,8,9)$, are depicted in Figure~\ref{fig:graphs}.

For each side-information graph, we develop a lower bound as in Lemma~\ref{cor_dof}. We recall that the lemma applies to acyclic subgraphs, which we create as follows;
for each permutation\footnote{We caution the reader not to confuse the current permutations ($\sigma$) that are used to construct large-sized acyclic graphs, with the aforementioned permutations $\pi$ which are used to construct $\mathcal{D}_{\Lc}$.} $\sigma\in S_3$, a set of nodes forming an acyclic subgraph is
\begin{align*}
&\{W^{\boldsymbol{d_{\sigma(1)}}(j)}_{\Tau_{1}}\}_{j=1}^{|\mathcal{U}_{\sigma(1)}|}\;\text{ for all}\; \Tau_{1}\subseteq \{1,2,3\}\setminus{\{\sigma(1)\}},\\
&\{W^{\boldsymbol{d_{\sigma(2)}}(j)}_{\Tau_{2}}\}_{j=1}^{|\mathcal{U}_{\sigma(2)}|} \;\text{ for all}\; \Tau_{2}\subseteq \{1,2,3\}\setminus{\{\sigma(1),\sigma(2)\}},\\
&\{W^{\boldsymbol{d_{\sigma(3)}}(j)}_{\Tau_{3}}\}_{j=1}^{|\mathcal{U}_{\sigma(3)}|} \;\text{ for all}\; \Tau_{3}\subseteq \{1,2,3\}\setminus{\{\sigma(1),\sigma(2),\sigma(3)\}}.
\end{align*}
Based on this construction of acyclic graphs, our task now is to choose a permutation $\sigma_s\in S_3$ that forms the maximum-sized acyclic subgraph. For the case where $\boldsymbol{d_1}=(8,9),\boldsymbol{d_2}=(1,2,3,4)$ and $\boldsymbol{d_3}=(5,6,7)$, it can be easily verified that such a permutation $\sigma_s$ is the one with $\sigma_s(1)=2$,$\sigma_s(2)=3$ and $\sigma_s(3)=1$.
In Figure~\ref{fig:graphs}, for each of the six graphs, we underline the nodes corresponding to the acyclic subgraph that is formed by such permutation $\sigma_s$. The outer bound now involves adding the sizes of these chosen (underlined) nodes.
For example, for the demand $\dv(\mathcal{U})=((8,9),(1,2,3,4),(5,6,7))$ (this corresponds to the lower center graph), the lower bound in~\eqref{eq:indexbound} becomes
\begin{align}
 T(\dv(\mathcal{U}))&\geq \frac{1}{2}\left(|W^{1}_{\emptyset}|+|W^{1}_{1}|+|W^{1}_{3}|+|W^{1}_{13}|+|W^{2}_{\emptyset}|\right.\nonumber\\
 &+|W^{2}_{1}|+|W^{2}_{3}|+|W^{2}_{13}|+|W^{3}_{\emptyset}|+|W^{3}_{1}|\nonumber\\
 &+|W^{3}_{3}|+|W^{3}_{13}|+|W^{4}_{\emptyset}|+|W^{4}_{1}|+|W^{4}_{3}|\nonumber\\
 &+|W^{4}_{13}|+|W^{5}_{\emptyset}|+|W^{5}_1|+|W^6_{\emptyset}|+|W^6_1|\nonumber\\
&\left.+|W^7_{\emptyset}|+|W^7_1|+|W^{8}_{\emptyset}|+|W^{9}_{\emptyset}|\right).
\end{align}
The lower bounds for the remaining $5$ vectors $\dv(\mathcal{U})$ for the same $\dv=(1,2,3,4,5,6,7,8,9)$, are given in a similar way, again by adding the (underlined) nodes of the corresponding acyclic subgraphs (again see Figure~\ref{fig:graphs}).

Subsequently, the procedure is repeated for all $P(N,K)=K!=9!$ worst-case demand vectors $\dv\in \mathcal{D}_{wc}$. Finally, all the $P(N,K)\cdot\Lambda!=9!\cdot 3!$ bounds are averaged to get
\begin{align} \label{eq:example_final}
&T(\Lc,\chi) \geq  \frac{1}{2}\frac{1}{9!\cdot3!}\nonumber \\
&\!\!\! \sum_{\dv(\mathcal{U})\in \mathcal{D}_{\Lc}}\sum_{\lambda\in[3]}\sum_{j=1}^{\mathcal{L}_{\lambda}}\sum_{\Tau_{\lambda}\subseteq [3]\setminus \{\sigma_s(1),\dots,\sigma_s(\lambda)\}}\!\!\!\!\!\!\!\!\!\!|W^{\boldsymbol{d_{\boldsymbol{\sigma_s(\lambda)}}}(j)}_{\Tau_{\lambda}}|
\end{align}
which is rewritten as
\begin{align}\label{eq:example_step1}
&T(\Lc,\chi) \geq  \frac{1}{2}\frac{1}{9!\cdot3!} \nonumber \\
&\sum_{i=0}^{3}\sum_{n\in[9]}\sum_{\Tau\subseteq[3]:|\Tau|=i} |W^n_{\Tau}| \cdot \underbrace{\sum_{\dv(\mathcal{U})\in \mathcal{D}_{\Lc}} \mathds{1}_{\mathcal{V}_{\mathcal{J}_s^{\dv(\mathcal{U})}}}(W^n_{\Tau})}_{Q_{i}(W^n_\Tau)}.
\end{align}
After the evaluation of the term $Q_i(W^n_\Tau)$, the bound in \eqref{eq:example_step1} can be written in a more compact form as
\begin{align}\label{eq:exbound}
T(\Lc,\chi)& \geq  \frac{1}{2}\sum_{i=0}^{3}\frac{\sum_{r=1}^{3-i}\mathcal{L}_r{3-r\choose i}}{9{3\choose i}}x_{i}\\
&\geq Conv\Bigg( \frac{1}{2}\frac{\sum_{r=1}^{3-i}\mathcal{L}_r{3-r\choose i}}{{3\choose i}}\Bigg)\label{eq:exbound2}
\end{align}
where the proof of the transition from (\ref{eq:example_step1}) to (\ref{eq:exbound}) and from (\ref{eq:exbound}) to (\ref{eq:exbound2}) can be found in the general proof (Section~\ref{sec:converse}).


\section{Conclusions\label{sec:discussion}}
We have treated the multi-sender coded caching problem with shared caches which can be seen as an information-theoretically simplified representation of some instances of the so-called cache-aided heterogeneous networks, where one or more transmitters communicate to a set of users, with the assistance of smaller nodes that can serve as caches.

The work is among the first --- after the work in \cite{WanTP15} --- to employ index coding as a means of providing (in this case, exact) outer bounds for more involved cache-aided network topologies that better capture aspects of cache-aided wireless networks, such as having shared caches and a variety of user-to-cache association profiles. Dealing with such non uniform profiles, raises interesting challenges in redesigning converse bounds as well as redesigning coded caching which is known to generally thrive on uniformity. Our effort also applied to the related problem of coded caching with multiple file requests.
In addition to crisply quantifying the (adverse) effects of user-to-cache association non-uniformity, the work also revealed a multiplicative relationship between multiplexing gain and cache redundancy, thus providing further evidence of the powerful impact of jointly introducing a modest number of antennas and a modest number of helper nodes that serve as caches. We believe that the result can also be useful in providing guiding principles on how to assign shared caches to different users, especially in the presence of multiple senders. Finally we believe that the current presented adaptation of the outer bound technique to non-uniform settings may also be useful in analyzing different applications like distributed computing \cite{LiAliAvestimehrComputIT18,PLE:18a,KonstantinidisRamamoorthyArxiv18,YanYangWiggerArxiv18,MingyueJiISIT18} or data shuffling \cite{AttiaTandon16,AttiaTandonISIT18,WanTuninettiShuffling18,MohajerISIT18} which can naturally entail such non uniformities.


\section{Appendix \label{sec:Appendix}}


\subsection{Proof of Lemma \ref{cor_dof} \label{proof:cor_dof}}
In the addressed problem, we consider a MISO broadcast channel with $N_0$ antennas at the transmitter serving $K$ receivers with some side information due to caches. In the wired setting this (high-SNR setting) is equivalent to the distributed index coding problem with $N_0$ senders $J_1,\dots,J_{N_0}$, all having knowledge of the entire set of messages, and each being connected via an (independent) broadcast line link of capacity $C_{J_i}=1, i\in[N_0]$ to the $K$ receivers which hold side information. This multi-sender index coding problem is addressed in~\cite{li2017cooperative}. By adapting the achievable rate result in \cite[Corollary1]{li2017cooperative} to our problem, we get
\begin{equation}
\sum_{\smallV\in \mathcal{V}_{\mathcal{J}}}R_{\smallV}\leq \sum_{i\in[N_0]}C_{J_i}
\end{equation}
($R_{\smallV} = \frac{|\smallV|}{T}$ is the rate for message $\smallV$), that yields
\begin{equation}\label{eq:lemma_1}
\sum_{\smallV\in \mathcal{V}_{\mathcal{J}}}\frac{|\smallV|}{T}\leq N_0
\end{equation}
which, when inverted, gives the bound in Lemma~\ref{cor_dof}. \qed

\subsection{Proof of Lemma \ref{lem:cons_acyclic} \label{proof:cons_acyclic}}
Consider a permutation $\sigma$ where the subfiles $W^{\boldsymbol{d_{\sigma(\lambda)}}(j)}_{\Tau_\lambda},\forall j\in\mathcal{U}_{\sigma(\lambda)}$ for all $\Tau_\lambda\subseteq[\Lambda]\setminus\{\sigma(1),\dots,\sigma(\lambda)\}$ are all placed in row $\lambda$ of a matrix whose rows are labeled by $\lambda = 1,2,\dots,\Lambda$. The index coding users corresponding to subfiles in row $\lambda$ only know (as side information) subfiles $W^{d_k}_{\Tau}, \ \Tau\ni \sigma(\lambda)$. Consequently each user/node of row $\lambda$ does not know any of the subfiles in the same row\footnote{Notice that the index coding users/nodes who are associated to the same cache, are not linked by any edge in the corresponding graph.} nor in the previous rows. As a result, the proposed set of subfiles chosen according to permutation $\sigma$, forms a subgraph that does not contain any cycle.

A basic counting argument can tell us that the number of subfiles --- in the acyclic subgraph formed by any permutation $\sigma\in S_{\Lambda}$ --- that are stored in exactly $i$ caches, is
\begin{equation}\label{eq:no_subfiles}
\sum_{r=1}^{\Lambda-i}|\Uc_{\sigma(r)}|{\Lambda-r\choose i}.
\end{equation} This means that the total number of subfiles in the acyclic subgraph is simply
\begin{equation}
\sum_{i=0}^{\Lambda}\sum_{r=1}^{\Lambda-i}|\Uc_{\sigma(r)}|{\Lambda-r\choose i}.
\end{equation}
This number is maximized when the permutation $\sigma$ guarantees that the vector $(|\Uc_{\sigma(1)}|,|\Uc_{\sigma(2)}|,\dots,|\Uc_{\sigma(\Lambda)}|)$ is in descending order. This maximization is achieved with our choice of the ordering permutation $\sigma_s$ (as this was defined in the notation part) when constructing the acyclic graphs.
 \qed

\subsection{Proof of Equation (\ref{eq:Qi}) \label{proof:lemmaQi}}
Here, through a combinatorial argument, we derive $Q_i(W^n_\Tau)$, that is the number of times that a subfile $W^n_\Tau$ with index size $|\Tau|=i$ appears in all the acyclic subgraphs chosen to develop the lower bound.

There are ${N-1\choose K-1}$ subsets $\Upsilon_{m},m\in[{N-1\choose K-1}]$ out of $N\choose K$ unordered subsets of $K$ files from the set $\{W^{j},j\in[N]\}$ that contain file $W^{n}$, and for each $\Upsilon_{m}$ there exists $K!$ different demand vectors $\dv'$. For each $\Upsilon_{m}$, among all possible demand vectors, a subfile $W^{n}_{\Tau}: |\Tau|=i$ appears in the side information graph an equal number of times. For a fixed $\Upsilon_{m}$, file $W^{n}$ is requested by a user connected to any helper node with a certain cardinality $\mathcal{L}_r$. 
By construction, $Q_{i}(W^n_\Tau)$ can be rewritten as
\begin{align}
Q_{i}(W^n_\Tau)&=\sum_{\dv\in \mathcal{D}_{wc}}\sum_{\pi\in S_{\Lambda}} \mathds{1}_{{\mathcal{V}_{\mathcal{J}_{s}^{\dv_r(\Uc)}}}}(W^{n}_{\Tau})\notag\\
&={N-1 \choose K-1}\sum_{r=1}^{\Lambda}\sum_{\dv'_{r}\in \mathcal{D}_{wc}}\sum_{\pi\in S_{\Lambda}} \mathds{1}_{{\mathcal{V}_{\mathcal{J}_{s}^{\dv'_r(\Uc)}}}}(W^{n}_{\Tau})\label{Qi} \notag
\end{align} where $\dv'_{r}$ denotes the subset of all demand vectors from $\Upsilon_{m}$ such that $n\in \boldsymbol{d_\lambda}:|\boldsymbol{d_\lambda}|=\mathcal{L}_r$.
The number of chosen maximum acyclic subgraphs containing $W^{n}_{\Tau}$ that arise from all the demand vectors $\dv'_{r}(\Uc)$ is evaluated as follows.
After fixing the demands such that $n\in\boldsymbol{d_\lambda}:|\boldsymbol{d_\lambda}|=\mathcal{L}_r $, then $W^n_\Tau$ appears in the side information graph only if it is requested by a user connected to helper node $\lambda$ such that $\lambda\notin\Tau$, which corresponds to $(\Lambda-i)$ different \textit{available} positions in the demand vector $\boldsymbol{d'}_{r}(\Uc)$, since $|\Tau|=i$.
After fixing one of the $(\Lambda-i)$ positions occupied by $\boldsymbol{d_\lambda}:|\boldsymbol{d_\lambda}|=\mathcal{L}_r$, for the remaining demands ${\boldsymbol{d_\lambda}}:|\boldsymbol{d_\lambda}|=\mathcal{L}_j,\forall j\in[\Lambda]\setminus\{r\}$ there are $P(\Lambda-i-1,r-1)\cdot(\Lambda-r)!$ possible ways to be placed into $\dv$. After fixing the order of $\boldsymbol{d_\lambda},\forall \lambda\in [\Lambda]$ in $\dv$ and $n\in \boldsymbol{d_\lambda}:|\boldsymbol{d_\lambda}|=\mathcal{L}_r$, there are $\mathcal{L}_r$ different positions in which $n$ can be placed in $\boldsymbol{d_\lambda}:|\boldsymbol{d_\lambda}|=\mathcal{L}_r$. This leaves out $\mathcal{L}_{r}-1$ positions with $K-1$ different numbers from the considered set $\Upsilon_{m}\setminus{\{n\}}$, and the remaining $K-\mathcal{L}_r$ positions in $\dv'_{r}$ are filled with $K-\mathcal{L}_r$ numbers. 
Therefore, there exist $\mathcal{L}_rP(K-1,\mathcal{L}_r-1)(K-\mathcal{L}_r)!$ different demand vectors where the subfile $W^{n}_{\Tau}$ will appear in the associated maximum acyclic subgraphs.
Hence, the above jointly tell us that
\begin{align}
&Q_{i}(W^n_\Tau)={N-1 \choose K-1}\sum_{r=1}^{\Lambda}P(\Lambda-i-1,r-1)\nonumber\\
&\times(\Lambda-r)!\mathcal{L}_rP(K-1,\mathcal{L}_r-1)(K-\mathcal{L}_r)!(\Lambda-i)
\end{align}
which concludes the proof. \qed

\subsection{Transition from Equation (\ref{eq:compacteq}) to (\ref{eq:LBwithxi}) \label{proof:transition}}
The coefficient of $x_i$ in equation (\ref{eq:compacteq}), can be further simplified as follows
\begin{align*} \label{eq:finanumber}
&\frac{Q_{i}}{\Lambda!P(N,K)} \\
=&\frac{(N-1)!(N-K)!}{(K-1)!(N-K)!\Lambda!N!}\sum_{r=1}^{\Lambda}\mathcal{L}_rP(K-1,\mathcal{L}_r-1)
\\ &(K-\mathcal{L}_r)!(\Lambda-i)P(\Lambda-i-1,r-1)(\Lambda-r)!\nonumber\\
=&\frac{1}{(K-1)!\Lambda!N}\sum_{r=1}^{\Lambda}\mathcal{L}_r\nonumber\\ &\frac{(K-1)!(K-\mathcal{L}_r)!(\Lambda-i)(\Lambda-i-1)!(\Lambda-r)!}{(K-\mathcal{L}_r)!(\Lambda-i-r)!}\nonumber\\
=&\frac{1}{\Lambda!N}\sum_{r=1}^{\Lambda}\mathcal{L}_r\frac{(K-1)!(\Lambda-i)!(\Lambda-r)!}{(K-1)!(\Lambda-i-r)!}\nonumber\\
=&\frac{1}{N}\sum_{r=1}^{\Lambda}L_{\pi_s(r)}\frac{(\Lambda-i)!(\Lambda-r)!i!}{\Lambda!(\Lambda-i-r)!i!}\nonumber\\
=&\frac{1}{N}\sum_{r=1}^{\Lambda}\mathcal{L}_r\frac{{\Lambda-r\choose i}}{{\Lambda\choose i}}
\end{align*}
which concludes the proof. \qed

\subsection{Monotonicity of $\{c_i\}$ \label{sec:monotonicity}}
Let us define the following sequences
\begin{eqnarray}
(a_n)_{n\in[\Lambda-i]}&\dfn& \bigg\{\frac{{{\Lambda-n}\choose i}}{{\Lambda \choose i}}, n\in [\Lambda-i]\bigg\} \\
(b_n)_{n\in[\Lambda-i-1]}&\dfn&\bigg\{\frac{{{\Lambda-n}\choose {i+1}}}{{\Lambda \choose {i+1}}}, n\in [\Lambda-i-1]\bigg\}.
\end{eqnarray}
 It is easy to verify that $a_n\geq b_n, \; \forall n\in [\Lambda-i]$. Consider now the set of scalar numbers $\{V_j, j\in[\Lambda-i], V_j\in \mathbb{N}\}$. The inequality $a^{*}_n\geq b^{*}_n, \; \forall n\in [\Lambda-i]$ holds for
\begin{equation}
{(a^{*}_n)_{n\in[\Lambda-i]}\dfn\big\{V_n\cdot a_n, n\in [\Lambda-i]\big\}}
\end{equation} and
\begin{equation}
{(b^{*}_n)_{n\in[\Lambda-i-1]}\dfn\big\{V_n\cdot b_n, n\in [\Lambda-i-1]\big\}}.
\end{equation} As a result, we have
\begin{equation}
\sum_{n\in[\Lambda-i]}V_n\cdot a_n \geq  \sum_{n\in[\Lambda-i]}V_n\cdot b_n \label{ineq:Lnan_Lnbn}
\end{equation}
which proves that $c_i\geq c_{i+1}$. \qed

\subsection{Proof of (\ref{eq:optimization1}) \label{lastproof}}

Through the respective change of variables  $t\dfn \Lambda\frac{M}{N}$, $x'_{i}\dfn \frac{x_{i}}{N}$ and $c'_{i}\dfn \frac{c_i}{N_0}$, in equations  (\ref{eq:LBwithxi}), (\ref{eq:sumfiles}) and (\ref{eq:constr2}), we obtain
\begin{eqnarray}
T(\Lc,\chi)&\geq &\sum_{i=0}^{\Lambda}x'_{i}c'_i \label{eq:compact2}\\
\sum_{i=0}^{\Lambda}x'_{i}&=&1 \label{eq:constr111}\\
\sum_{i=0}^{\Lambda}ix'_{i}&\leq& t .\label{eq:constr22}
\end{eqnarray}
Let $X$ denote a discrete integer-valued random variable with probability mass function $f_{X}(x)=\{x'_i ~\text{if}~x=i, \forall i\in \{0,1,\dots,\Lambda\}\}$, where the $x'_i$ are those that satisfy equation \eqref{eq:constr111}. The value $c'_i$ can also be seen as the realization of a random variable $Y\dfn g(X)$, where $g(x)=\frac{\sum_{r=1}^{\Lambda-x}\mathcal{L}_r{\Lambda-r\choose x}}{N_0{\Lambda\choose x}}$, having the same probability mass function as $X$, i.e. $f_{Y}(y)=\{x'_i ~\text{if}~y=c'_i, \forall i\in \{0,1,\dots,\Lambda\}\}$. Due to the equation in (\ref{eq:constr22}), the expectation of $X$ is bounded as $\mathbb{E}[X]\leq t$. Similarly, (\ref{eq:compact2}) is equivalent to $T(\Lc,\chi)\geq \mathbb{E}[Y]$. From Jensen's inequality, we have
$
T(\Lc,\chi)\geq \mathbb{E}[Y]\geq g(\mathbb{E}[X])
$.
Since the sequence $\{c'_i\}$ (and equivalently the function $g(x)$) is monotonically decreasing, the following lower bound holds
\begin{equation}
T(\Lc,\chi)\geq g(\mathbb{E}[X])\geq g(t)=\frac{\sum_{r=1}^{\Lambda-t}\mathcal{L}_r{\Lambda-r\choose t}}{N_0{\Lambda\choose t}}.
\end{equation}
This concludes the proof. \qed

\subsection{Proof of Equation (\ref{eq:totdelay2}) \label{sec:BinomialChangeProof}}
We remind the reader that (for brevity of exposition, and without loss of generality) this part assumes that the $|\Uc_{\lambda}|$ are in decreasing order.

We define the following quantity
\begin{equation*}
b_\lambda\dfn |\Uc_{1}|-|\Uc_{\lambda}|
\end{equation*}
and rewrite the total number of transmissions using the above definition as
\begin{align*}
&\sum_{j=1}^{|\Uc_{1}|}{{\Lambda \choose \Lambda\gamma+1}-{a_j \choose \Lambda\gamma+1}}\notag \allowbreak\\
&=|\Uc_{1}|{\Lambda \choose \Lambda\gamma+1}-\sum_{j=1}^{|\Uc_{1}|}{{a_j \choose \Lambda\gamma+1}}\notag \allowbreak\\
&=\sum_{i=1}^{\Lambda-\Lambda\gamma}{(|\Uc_{i}|+b_i){\Lambda-i \choose \Lambda\gamma}}-\sum_{j=1}^{|\Uc_{1}|}\sum_{i=\Lambda\gamma}^{a_j-1}{{i \choose \Lambda\gamma}}\allowbreak \notag \\
&=\sum_{i=1}^{\Lambda-\Lambda\gamma}{|\Uc_{i}|{\Lambda-i \choose \Lambda\gamma}}+{\sum_{i=\Lambda\gamma}^{\Lambda-1}{b_{\Lambda-i}{i \choose \Lambda\gamma}}}\allowbreak \notag\\
&-{\sum_{j=1}^{|\Uc_{1}|}\sum_{i=\Lambda\gamma}^{a_j-1}{{i \choose \Lambda\gamma}}}\allowbreak\notag
\end{align*}
\begin{align}
&\overset{(a)}{=}\sum_{i=1}^{\Lambda-\Lambda\gamma}{|\Uc_{i}|{\Lambda-i \choose \Lambda\gamma}}+{\sum_{i=\Lambda\gamma}^{\Lambda-1}{\sum_{j:a_j\geq i+1}^{|\Uc_{1}|}{i \choose \Lambda\gamma}}}\notag\allowbreak\\
&-{\sum_{j:a_j-1\geq \Lambda\gamma}^{|\Uc_{1}|}\sum_{i=\Lambda\gamma}^{a_j-1}{{i \choose \Lambda\gamma}}}\notag\allowbreak\\
&\overset{(b)}{=}\sum_{i=1}^{\Lambda-\Lambda\gamma}{|\Uc_{i}|{\Lambda-i \choose \Lambda\gamma}}+\sum_{j:a_j\geq \Lambda\gamma+1}^{|\Uc_{1}|}{\sum_{i=\Lambda\gamma}^{a_j-1}{i \choose \Lambda\gamma}}\notag \allowbreak\\
&-\sum_{j:a_j-1\geq \Lambda\gamma}^{|\Uc_{1}|}\sum_{i=\Lambda\gamma}^{a_j-1}{{i \choose \Lambda\gamma}}\notag\\
&=\sum_{i=1}^{\Lambda-\Lambda\gamma}{|\Uc_{i}|{\Lambda-i \choose \Lambda\gamma}}\label{eq:finalform}
\end{align}
where step $(a)$ uses the equality $b_{\Lambda-i}=\sum_{j:a_j\geq i+1}^{|\Uc_{1}|}{1}$, and where step $(b)$ follows by changing the counting order of the double summation in the second summand. Substituting (\ref{eq:finalform}) into the numerator of (\ref{eq:totdelay1}) yields the overall delivery time given in (\ref{eq:totdelay2}).
The same performance holds for any $\Uc$ with the same profile $\Lc$.
\qed

\subsection{Transition to the Multiple File Request Problem\label{sec:AppendixMultipleFileRequests}}
We here briefly describe how the converse and the scheme presented in the shared cache problem, can fit the multiple file request problem.
\paragraph*{Converse}
In Remark~\ref{rem:multipleFilerequstsResult} we described the equivalence between the two problems. Based on this equivalence, we will describe how the proof of the converse in Section~\ref{sec:converse} holds in the multiple file request problem with $N_0=1$, where now simply some terms carry a different meaning. Firstly, each entry $\boldsymbol{d_{\lambda}}$ of the vector defined in equation \eqref{eq:OrderDemand2} now denotes the vector of file indices requested by user $\lambda$. Then we see that Lemma~\ref{lem:cons_acyclic} (proved in Section~\ref{proof:cons_acyclic}) directly applies to the equivalent index coding problem of the multiple file requests problem, where now, for a given permutation $\sigma$ (see Section~\ref{proof:cons_acyclic}), all the subfiles placed in row $\lambda$  --- i.e., subfiles $W^{\boldsymbol{d_{\sigma(\lambda)}}(j)}_{\Tau_\lambda},\forall j\in\mathcal{U}_{\sigma(\lambda)}$ for all $\Tau_\lambda\subseteq[\Lambda]\setminus\{\sigma(1),\dots,\sigma(\lambda)\}$ --- are obtained from different files requested by the same user, and therefore any two of these subfiles/nodes are not connected by any edge in the side information graph. After these two considerations, the rest of the proof of Lemma~\ref{lem:cons_acyclic} is exactly the same. The remaining of the converse consists only of mathematical manipulations which remain unchanged and which yield the same lower bound expression.

\paragraph*{Scheme}
The cache placement phase is identical to the one described in
Section~\ref{sec:SchemePlacement}, where now each cache $\lambda$ is associated to the single user $\lambda$.
In the delivery phase, the scheme now follows directly the steps in
Section~\ref{sec:SchemeDelivery} applied to the shared-link (single antenna) setting, where now $\mathcal{A}_{\lambda}=\mathcal{U}_{\lambda}$ (cf.~\eqref{eq:Alambda}). As in the case with shared caches, the scheme consists of $\mathcal{L}_1$ rounds, each serving users
\begin{equation}\label{eq:UsersServerPerRound2}
\mathcal{R}_j=\bigcup_{\lambda\in[\Lambda]} \big( \mathcal{U}_{\lambda}(j):\mathcal{L}_\lambda \geq j \big)
\end{equation} where $\mathcal{U}_{\lambda}(j)$ is the $j$-th user in set $\mathcal{U}_{\lambda}$. The expression in~\eqref{eq:UsersServerPerRound2} now means that the multiple files requested by each user are transmitted in a \emph{time-sharing} manner, and at each round the transmitter serves at most one file per user. Next, equation \eqref{eq:UsersServedPerXOR} is replaced by
\begin{equation}
\chi_\mathcal{Q}=\bigcup_{\lambda\in \mathcal{Q}}\big( \mathcal{U}_{\lambda}(j):\mathcal{L}_\lambda \geq j \big)
\end{equation}
and then each transmitted vector described in equation~\eqref{eq:TransmitSignalGeneral}, is substituted by the scalar
\begin{equation}
x_{\chi_{\mathcal{Q}}}=\!\!\!\!\bigoplus_{\lambda\in \mathcal{Q}:\mathcal{L}_\lambda \geq j} W^{d_{\mathcal{U}_{\lambda}(j)}}_{\mathcal{Q}\backslash{\{\lambda\}},1}.
 \end{equation}
Finally decoding remains the same, and the calculation of delay follows directly.

\bibliography{final_refs2}
\bibliographystyle{IEEEtran}
\end{document}